\documentclass[11pt,letterpaper]{article}

\usepackage{amsmath,amssymb,amsthm}
\usepackage{thmtools}
\usepackage{mathtools}
\usepackage{graphicx} 
\usepackage[margin=1in]{geometry}
\usepackage{xcolor}
\definecolor{darkgreen}{rgb}{0,0.5,0}
\usepackage[colorlinks=true, citecolor=darkgreen,backref=page]{hyperref}

\usepackage[textsize=footnotesize]{todonotes}
\usepackage{setspace}
\usepackage{microtype}
\usepackage{enumitem}
\usepackage{subcaption}
\usetikzlibrary{backgrounds,calc}
\usepackage{booktabs}
\usepackage{multirow} 
\usepackage{hhline}
\usepackage{authblk}
\usepackage[ruled,vlined]{algorithm2e}

\usepackage{stmaryrd}

\usepackage[capitalise]{cleveref}

\usetikzlibrary{calc}

\def\CongTest{\textsf{\textup{CongTest}}}
\def\CongIden{\textsf{\textup{CongIden}}}
\def\ApxCong{\textsf{\textup{ApxCong}}}

\def\EQ{\textsf{\textup{EQ}}}
\def\DISJ{\textsf{\textup{DISJ}}}
\def\CongSign{\textsf{\textup{CongSign}}}
\def\CongHash{\textsf{\textup{CongHash}}}

\def\dbR{\mathbb{R}}
\def\dbQ{\mathbb{Q}}
\def\dbN{\mathbb{N}}
\def\dbZ{\mathbb{Z}}
\def\dbC{\mathbb{C}}
\def\dbF{\mathbb{F}}

\def\calB{\mathcal{S}}
\def\BB{\mathcal{B}}
\def\calP{\mathcal{P}}
\def\calF{\mathcal{F}}

\newcommand{\set}[1]{\{#1\}}
\newcommand{\setB}[1]{\left\{#1\right\}}
\newcommand{\abs}[1]{|#1|}
\newcommand{\absB}[1]{\left|#1\right|}
\newcommand{\norm}[1]{\|#1\|}

\def\defeq{\coloneqq}
\def\eqdef{\eqqcolon}

\def\Dcc{\textsf{\textup{D}}^{\textit{cc}}}
\def\Rcc{\textsf{\textup{R}}^{\textit{cc}}}

\theoremstyle{plain}
\newtheorem{theorem}{Theorem}[section]
\newtheorem{corollary}[theorem]{Corollary}
\newtheorem{proposition}[theorem]{Proposition}
\newtheorem{lemma}[theorem]{Lemma}
\newtheorem{fact}[theorem]{Fact}
\newtheorem{claim}[theorem]{Claim}
\newtheorem{observation}[theorem]{Observation}

\theoremstyle{definition}
\newtheorem{definition}[theorem]{Definition}

\def\ii{\mathbf{i}}
\def\bfe{\mathbf{e}}
\newcommand{\boundedQ}[1]{\dbQ_{\langle#1\rangle}}

\SetCommentSty{commfont}
\SetKwComment{Comment}{\textcolor{green!50!black}{$\blacktriangleright$~}}{}

\definecolor{purple}{HTML}{1A1AB3}
\definecolor{darkgreen}{HTML}{00AA00}

\DeclareEmphSequence{\color{purple}\itshape}

\AtBeginEnvironment{algorithm}{%
  \DeclareEmphSequence{\itshape}
}
\AtEndEnvironment{algorithm}{%
  \DeclareEmphSequence{\color{purple}\itshape}
}

\AtBeginEnvironment{thebibliography}{%
  \DeclareEmphSequence{\itshape}
}
\AtEndEnvironment{thebibliography}{%
  \DeclareEmphSequence{\color{purple}\itshape}
}

\newcommand{\Modp}[2][p]{\llbracket#2\rrbracket_{#1}}

\newcommand{\ModFpi}[2][p]{\llbracket#2\rrbracket_{\dbF_{#1}[\ii]}}

\newcommand{\pequiv}[1][p]{\equiv_{#1}}

\DeclareMathOperator{\Var}{Var}
\DeclareMathOperator{\E}{\mathbb{E}}

\def\poly{\textup{\textsf{poly}}}
\def\polylog{\textup{\textsf{polylog}}}

\def\imageat{\includegraphics[scale=0.0234]{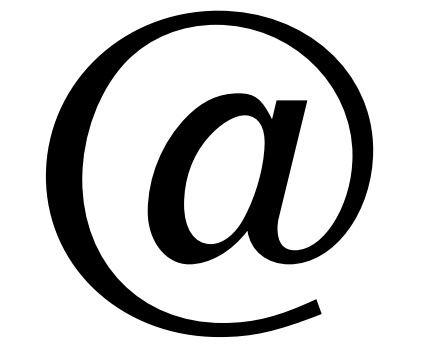}}
\def\imagedot{\includegraphics[scale=0.0234]{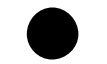}}

\SetKwComment{fulltcp}{$\blacktriangleright$ }{}
\SetKwComment{sectcp}{// }{}
\SetKwComment{tcp}{$\blacktriangleleft$ }{}
\SetKwInput{KwParameter}{Parameter}
\SetKw{KwAbort}{\textcolor{darkgreen}{Abort}}
\SetKw{KwReturn}{Return}

\DeclareMathOperator{\Li}{Li}

\newcommand{\floor}[1]{\lfloor#1\rfloor}

\title{Efficient Streaming Algorithms for Two-Dimensional Congruence Testing and Congruence Hashing}

\date{}

\author[(1)]{Yen-Cheng Chang \thanks{chen111062204\imageat{}gapp\imagedot{}nthu\imagedot{}edu\imagedot{}tw}}

\author[(1)]{Tsun-Ming Cheung \thanks{bencheung\imageat{}cs\imagedot{}nthu\imagedot{}edu\imagedot{}tw}}

\author[(2)]{Meng-Tsung Tsai \thanks{mttsai\imageat{}iis\imagedot{}sinica\imagedot{}edu\imagedot{}tw}}

\author[(2)]{Ting-An Wu \thanks{ericwu\imageat{}iis\imagedot{}sinica\imagedot{}edu\imagedot{}tw}}

\affil[(1)]{Department of Computer Science, National Tsing Hua University}
\affil[(2)]{Institute of Information Science,
Academia Sinica}

\begin{document}
\maketitle
\thispagestyle{empty}

\begin{abstract} {
The \emph{geometric congruence} problem is a fundamental building block in many computer vision and image recognition tasks. This problem asks whether two point multisets are congruent under translation and rotation. 
A related and more general problem, \emph{congruence hashing}, considers the task of compactly encoding multiple point multisets for efficient congruence queries.
Despite its wide applications, both problems have received little prior attention in space-conscious settings. In this work, we study these problems in the \emph{streaming model}, where data arrive as a stream and the primary goal is to minimize the space usage.
To meaningfully analyze space complexity, we account for the underaddressed issue of input precision by working in the \emph{finite-precision rational} setting: the input point coordinates are rational numbers of the form $p/q$ with $|p|, |q| \le U$.

Our main algorithmic results are two randomized polylogarithmic-space algorithms for the stronger variant of \emph{congruence identification} problem: given two $n$-point $U$-rational multisets, determine whether they are congruent and, if so, output a valid transformation. 
We present two $3$-pass algorithms based on two different approaches. With probability at least $1-\delta$,  
\begin{itemize}
    \item Our \emph{product-anchor} algorithm uses $O\left((\log n+\log U+\log \tfrac{1}{\delta})\,\log\log n\cdot\log \tfrac{1}{\delta}\right)$ space in the turnstile model, and
    \item Our \emph{complex-moment} algorithm uses $O(\log n(\log n+\log U+\log \tfrac{1}{\delta}))$ space in the insertion-only model.
\end{itemize}
Using congruence identification as a building block, we obtain a $4$-pass $O(m(\log n+\log U+\log m+\log \tfrac{1}{\delta}))$-space insertion-only algorithm for congruence hashing over $m$ query sets $S_1,\dots,S_m\subseteq\mathbb{Q}^2$ with rational precision $U$ that outputs a congruence signature of length $O(\log \tfrac{1}{\delta}+\log U+\log m)$. 

On the technical side, both algorithms presample primes for modular hashing to overcome precision issues. The product-anchor algorithm combines finite-field embeddings with number-theoretic results to ensure that with high probability, one of the presampled moduli permits efficient rotation recovery. The complex-moment algorithm instead builds on a new structural result: for every nonzero $n$-point rational multiset, at least one of its $O(\log n)$ power-of-two complex moments is nonzero. Tracking only these moments therefore guarantees efficient rotation recovery without any additional number-theoretic condition. 
We believe this non-vanishing condition may be of independent interest, since complex moments have long been used heuristically in object recognition but are often limited by the vanishing moment phenomenon (Flusser and Suk, \text{[}IEEE Transactions on Image Processing\text{]}, 2006), and we overcome this obstacle in our setting by maintaining only a few carefully chosen complex moments with proven theoretical soundness.

Lastly, we complement our algorithmic results with the lower bound result matching the turnstile product-anchor algorithm up to $\log\log n$ factor: any $p$-pass randomized streaming algorithm for two-dimensional congruence identification with error at most $\delta$ requires $\Omega\left(\tfrac{1}{p}\left(\log n+\log U+\log \tfrac{1}{\delta}\right)\right)$ space. Moreover, we show that approximate congruence testing, which asks whether two point multisets can be aligned up to small perturbations, requires $\poly(n)$ space even with $\poly(n)$ passes.
This sharply contrasts with the RAM model, where both exact and approximate congruence testings are solvable in polynomial time.
}
\end{abstract}

\thispagestyle{empty}
\tableofcontents

\thispagestyle{empty}
\clearpage
\setcounter{page}{1}

\section{Introduction} \label{sec:Intro}
\emph{Geometric congruence} is a classical problem with broad use in various areas such as computer vision and image recognition. 
The basic setup consists of two point multisets $A,B\subseteq D^d$ of size $n$ each, where $D\subseteq \dbR$. Each point may occur with multiplicity in both sets.\footnote{Throughout this work, all inputs to geometric congruence problems are interpreted as multisets. We nevertheless occasionally use the more natural terminology of ``point sets''.}
Our main focus is the \emph{congruence identification} problem $\CongIden_{n,D,d}$: determine whether there exist a rotation $\rho:\dbR^d\to \dbR^d$ and a translation $t\in \dbR^d$ such that $a\mapsto \rho(a)+t$ is a multiplicity-preserving bijection from $A$ to $B$, and output a possible pair of $(\rho,t)$ in the congruent case.


Congruence identification is stronger than the more commonly studied decision problem of \emph{congruence testing} $\CongTest_{n,D,d}$.
Fixing the input size is without loss of generality, as the algorithm can track the cardinalities and reject if they differ. Allowing reflections also does not increase the asymptotic complexity: every such congruence has the form $B=\rho(\sigma(A))+t$, where $\sigma\in\set{I,J}$ for a fixed reflection $J$ \cite{AKUTSU1998247}, so two parallel runs of the algorithm can solve the problem.

We study the geometric congruence problem allowing multiplicities. We emphasize that in a streaming setting, a multiset is not equivalent to a pre-aggregated tally representation of points with integral weights unless identical points arrive consecutively.

Generalizing the congruence testing problem to multiple sets, \emph{congruence hashing} $\CongHash^m_{n,D,d}$ approaches from the perspective of data structure, with a focus on efficient query processing over multiple point sets. Formally, the input of the congruence hashing problem consists of $m$ point multisets (also called \emph{query sets}) $A_1,\ldots,A_m\subseteq D^d$ of size $n$ each.
The main objective is to output a \emph{congruence signature function} $H:[m]\to\set{0,1}^s$, such that with high probability, for any $i,j\in[m]$, $H(i)=H(j)$ if and only if $A_i$ and $A_j$ are congruent.
Apart from the (time or space) complexity required to generate the hash function, we also aim to minimize the \emph{hash length} $s$ to support efficient comparison queries.

\subsection{Geometric Congruence Algorithms in RAM Model}
The time complexity of the geometric congruence problem has been extensively studied in the RAM model. In two dimensions, Manacher \cite{MANACHER19766} gives an \(O(n\log n)\)-time algorithm by sorting the points by angle and performing a cyclic string matching. In three dimensions, Atkinson \cite{ATKINSON1987159} achieves the same \(O(n \log n)\) bound using the closest-pair structures. In four dimensions, Kim and Rote \cite{kim2016congruence} give an \(O( n\log n)\)-time algorithm. For $d\geq 4$, Akutsu \cite{AKUTSU1998247} devises an  \(O(n^{(d - 1)/2} \log n) \)-time randomized algorithm for $d$ dimensions based on the birthday paradox. 

Compared to congruence testing, congruence hashing has received less attention from the worst-case algorithmic viewpoint, but it is widely applied in the computer vision literature, see for example \cite{641604}. While closely related to congruence, congruence hashing places additional emphasis on \emph{hashing}: more focus is put on obtaining a short hash value that enables fast congruence checking across many instances.

For both the problems of congruence testing and congruence hashing, little attention has been paid to space complexity, leaving the efficiency of problems in space-sensitive settings underexplored. 
Most time-complexity analyses implicitly assume the \emph{Real-RAM model}, where each input value is stored in a single machine word with full precision. In contrast, in real-world applications, the input must be representable by finite strings.
One intuitive workaround is to consider an \(\epsilon\)-approximate variant of \(\CongIden\) where we allow two point sets to be considered congruent if they can be aligned with a perturbation error of \(\epsilon\). While this problem is solvable in polynomial time in the RAM model \cite{alt1988congruence}\footnote{An earlier version of this manuscript incorrectly claimed the problem to be NP-hard in the RAM model, citing \cite{IWANOWSKI1991227}; we correct that claim here.}, in \cref{thm:apx-cong-lower-bound} we establish a strong pass-space tradeoff for approximate congruence testing, thus ruling out efficient streaming algorithms in the sublinear-space regime.

Another natural modification is to consider finite-precision integers by taking the domain to be $[U]\defeq \set{1,2,\ldots,U}$ for some $U\in \dbN$.
We strengthen this setting by considering a substantially larger class of finite exact inputs, namely the finite-precision rational numbers. Concretely, for $U\in \dbN$, the set of \emph{$U$-rational numbers} is defined as
\(
    \boundedQ{U} \defeq \setB{\pm \frac{p}{q}:p,q\in [U]} \cup\set{0}.
\)
At first glance, finite-precision rational numbers may appear similar to finite-precision integers. However, the precision required for addition differs drastically while the sum of $n$ integers from $[U]$ can be exactly stored as an integer in $[nU]$, the exact sum of $n$ $U$-rational numbers may have to be represented as an $nU^n$-rational number (\cref{obs:sum-prod-precision}). 

\subsection{Streaming Model and Limitations of Prior RAM-Model Algorithms} \label{subsec:RAM-limitation}
In this work, we consider geometric congruence in the \emph{streaming model}. 
In this model, data arrive sequentially in an adversary order. A streaming algorithm scans the stream forward for one or few passes while using limited space for storage and computations. Two common variants are the \emph{insertion-only} model, where elements are only inserted, and the \emph{turnstile} model, where in addition inserted elements may be deleted. Turnstile model is stronger than insertion-only model, so a turnstile streaming algorithm is automatically an insertion-only algorithm, and conversely, an insertion-only streaming lower bound applies to turnstile model.

All the aforementioned RAM-model approaches do not extend efficiently to the streaming model. Any \(p\)-pass sorting algorithm for \(n\) elements requires \(\Omega(n/p)\) space~\cite{MUNRO1980315}, and even the simpler task of median estimation (i.e. returning a datum close to the median) necessitates $\poly(n)$ space for $O(\log n/\log\log n)$-pass algorithms \cite{MUNRO1980315,GM09}. Another RAM-model primitive of finding a closest pair of points requires $\Omega(n/p)$ space for any randomized $p$-pass algorithm \cite{Sharathkumar2007}, thus ruling out any efficient closest-pair based adaptations to the streaming model. Furthermore, Real-RAM model does not deal with the precision issues, the intermediate values used by RAM-model routines may not be representable to sufficient precision, which is the crux of the space complexity aspect of the streaming model. For example, for a $U$-rational point set $\set{z_1,\ldots,z_n}\subseteq \boundedQ{U}^2$, the centroid \(c=\frac{1}{n}\sum_{i=1}^n z_i\) may need \(\Theta(n\log U)\) bits of precision, so the routine centering step could already break the intended space bound if we record it directly. 

Recognizing the fundamental obstacle of RAM-model techniques in designing space-conscious streaming algorithms, we shift our attention to \emph{moment-based approaches}.
Moment-based methods were introduced to image analysis in the classical work of Hu \cite{Hu62}, who derived seven algebraic combinations of geometric moments that remain invariant under translation, rotation, and scaling. Since then, moment invariants have been applied in various recognition and matching tasks due to their computational simplicity and memory efficiency, including aircraft silhouette identification \cite{5009272}, scene matching \cite{WONG197816}, and template matching under rotation \cite{Goshtasby85}; see also the survey \cite{PROKOP1992438} and the references therein.

\subsection{Algorithmic Results}
As mentioned above, basic arithmetic on bounded-precision rationals can already cause substantial precision blowups. Our main result is that, despite these complications, congruence identification for rational points can be performed in space polylogarithmic in the input size $n$ and the precision bound $U$.

Our main algorithmic results for $\CongIden_{n,\boundedQ{U},2}$ are obtained by two independent methods, which we present in \cref{sec:PA-CongIden,sec:CM}.
\begin{theorem}[Product-anchor Algorithm]\label{thm:CongIden-PA}
    For $\delta\in (0,1)$, $\CongIden_{n,\boundedQ{U},2}$ can be solved by a 3-pass $O\left( \left(\log n+\log U+\log \frac{1}{\delta}\right) \log\log n \cdot \log \frac{1}{\delta}\right)$-bit randomized turnstile streaming algorithm with probability at least \(1 - \delta\). 
\end{theorem}
\begin{theorem}[Complex-moment Algorithm]\label{thm:main}
    For $\delta\in (0,1)$,
    $\CongIden_{n,\boundedQ{U},2}$ can be solved by a 3-pass $O(\log n(\log n+\log U+\log \frac{1}{\delta}))$-bit randomized insertion-only streaming algorithm with probability at least $1-\delta$.
\end{theorem}
The insertion-only complex-moment algorithm in \cref{thm:main} can be adapted to the turnstile model with the same pass complexity and a slightly increased polylogarithmic space bound, which we defer the details to \cref{cor:CM-turnstile}.

The comparison between the space bounds of the algorithms in \cref{thm:CongIden-PA,thm:main} depends crucially on the failure probability $\delta$ and the model. In the insertion-only model,
the product-anchor algorithm outperforms the complex-moment algorithm for constant $\delta$, while the situation is reversed when $\delta = 1/n^{\Omega(1)}$. The two algorithms have
asymptotically matching space bounds when
$\delta = n^{-\Theta(1/\log\log n)}$. 
In the turnstile model, the additional space required by the $\ell_0$-samplers in the complex-moment algorithm (\cref{cor:CM-turnstile}) makes the product-anchor algorithm typically more space-efficient. 
The different dependency on $\delta$ in the two space bounds arises primarily from the distinct prime-sampling procedures used by the two algorithms to hash rational points into finite fields. We elaborate on this distinction in \cref{sec:tech-overview}.

We note that the dependence on $U$ is fairly mild: the algorithm still uses polylogarithmic space even when the input precision is quasi-polynomial in $n$. 
The same algorithms also apply to finite-precision \emph{integer} inputs. By \cref{thm:CongIden-PA}, the congruence identification problem on $[\poly(n)]^2$ can be solved by a 3-pass $O(\log n \log\log n)$-space randomized turnstile streaming algorithm with constant success probability, which is nearly optimal since representing a valid transformation already requires $\Theta(\log n)$ space. 

For the special case of point sets, in which each point has multiplicity one in each dataset, our product-anchor algorithm eliminates an additional factor of $O(\log\log n\cdot \log\frac{1}{\delta})$ from the space usage. Together with the hardness results in \cref{sec:hardness}, this shows that the resulting space complexity is indeed tight.
\begin{theorem}\label{thm:PA-point-set}
    Suppose further that the input datasets are simple sets.
    For $\delta\in (0,1)$, $\CongIden_{n,\boundedQ{U},2}$ can be solved by a 3-pass $O\left(\log n+\log U+\log \frac{1}{\delta}\right)$-bit randomized turnstile streaming algorithm with probability at least \(1 - \delta\). 
\end{theorem}

For the congruence hashing on $m$ query sets over $\boundedQ{U}^2$, we obtain a 4-pass streaming algorithm that uses $\poly(m,\log n,\log U, \log \frac{1}{\delta})$ space and outputs a hash function of logarithmic length. Moreover, our congruence hashing algorithm ensures that for any pair of congruent sets, a valid transformation can be recovered from their signatures.
\begin{restatable}{theorem}{geohash}\label{thm:geomhash-database}
For $\delta\in (0,1)$, there exists a 4-pass randomized insertion-only streaming algorithm for $\CongHash_{n,\boundedQ{U},2}^m$ using $O(m (\log n + \log U + \log m + \log \frac{1}{\delta}))$ space, which outputs a congruence signature of length $O(\log U+\log m+\log\frac{1}{\delta})$ with probability at least $1-\delta$. 
\end{restatable}

Furthermore, inspired by \cite{AKUTSU1998247}, we develop a constant-pass streaming algorithm that solves $\CongIden_{n, \boundedQ{U}, 3}$ with sub-linear space usage. We give an overview of the limitations of extending two-dimensional techniques in $\boundedQ{U}^3$ in \cref{subsubsec:TO-3D} and discuss our three-dimensional algorithm in \cref{sec:3D}.

\subsection{Hardness results}\label{subsec:hardness-intro}
For the lower bound results, we show that the lower bound can be proved via a reduction from the communication problem of equality \cite{Rou16}, which is tight up to $\log\log$ factor.

\begin{corollary}\label{cor:2D-congiden-LB}
    Any $p$-pass randomized streaming algorithm for $\CongIden_{n,\boundedQ{U},2}$ with error at most $\delta\in (0,1/2)$ requires a space usage of $\Omega\left(\frac{1}{p}\left(\log n+\log U+\log \frac{1}{\delta}\right)\right)$.
\end{corollary}

This justifies that both precision bound and input size are relevant parameters linked to the complexity of geometric congruence problems.

One-pass hardness results for geometric problems in streaming are scarce. One of the most relevant results is the randomized linear sketch lower bound for homomorphic fingerprint \cite{10.1145/2488608.2488726}. The lower bound only holds for the class of shift-homomorphic linear sketches and thus implies no general lower bound on cyclic shift-related problems. Importantly, the string-rotation lower bound result does not apply to our setting: a cyclic shift of $n$ symbols corresponds to an order-$n$ rotation, but Niven's theorem (\cref{lem:niven's}) restricts rational planar rotations to orders $1,2$ or $4$, so no such embedding exists in the rational two-dimensional setting.

As mentioned before, approximate congruence ($\ApxCong$) is solvable in polynomial time in the RAM model \cite{alt1988congruence}. In the streaming setting, we show that as long as $U=\Omega(n^2)$, 
any \(p\)-pass randomized streaming algorithm for the \(\epsilon\)-approximate requires \(\Omega(n/p)\) bits of space for essentially every meaningful regime of the error tolerance $\epsilon$.

\begin{restatable}{theorem}{apxHardness}\label{thm:apx-cong-lower-bound}
    Let $n,d\in \dbN$. For any $\epsilon\in (1/U, U/10)$ and $\delta\in (0,1/2)$, any $p$-pass randomized streaming algorithm for $\ApxCong_{n,\boundedQ{U},d}^\epsilon$ with error $1/2-\delta$ requires $\Omega\left(\frac{\delta^2}{p}\min\setB{n, \frac{U}{\max\set{\epsilon,1}}}\right)$. 
\end{restatable}

Our main motivation for studying the exact finite-precision rational variant is to introduce a space constraint that renders the problem meaningful in the streaming model.
There are two natural ways to introduce space constraints to geometric congruence problems, namely restricting the input precision for exact congruence, and rounding the inputs and allowing a perturbation tolerance. \cref{thm:apx-cong-lower-bound} essentially rules out the latter approach, suggesting that the finite-precision assumption for exact congruence is not merely an artifact for streaming adaptation.

\paragraph{Paper Organization.}
The remainder of the paper is organized as follows. In \cref{sec:tech-overview}, we give a technical overview of our algorithms. \Cref{sec:prel} defines notation and presents the preliminaries. \Cref{sec:PA-CongIden,sec:CM} present the product-anchor algorithm and complex-moment algorithm for two-dimensional congruence identification respectively. \Cref{sec:geometric-hashing} presents the algorithm for congruence hashing. In \Cref{sec:3D}, we discuss our algorithm for three-dimensional congruence identification using birthday paradox. \Cref{sec:hardness} establishes the streaming lower bound results.

\section{Technical Overview}\label{sec:tech-overview}
We begin with a structural overview of our congruence-identification algorithms. Both the \emph{product-anchor} and \emph{complex-moment} algorithms follow a common three-stage framework. In the first pass, the algorithm recovers the centroids of the two point multisets, which determines the correct translation and reduces the remaining task to solely rotation identification. In the second pass, the algorithm maintains a small collection of streaming-friendly statistics from which a constant number of candidate rotations can be extracted offline. In the third pass, given these candidate rotation-translation pairs $(\rho,t)$, the algorithm tests multiset equality between $B$ and $\rho A+t$ via a constant number of parallel standard Karp-Rabin fingerprints \cite{5390135}.

Often in two-dimensional geometry, $(x,y)\in \dbR^2$ is identified with the complex number $z=x+y\ii\in \dbC$, so a counterclockwise rotation by $\theta$ of $(x,y)$ about the origin is succinctly represented by the complex number multiplication $e^{\ii\theta}z$. 
As we will see, the algebraic perspective not only enables efficient computation but also motivates our algebraic approaches to rotation extraction.

\subsection{Precision Issues}
Centroids are natural choices for pinning down the translation and thus simplifying the subsequent task to only finding a rotation. As we will show in \cref{app:recentring}, the centroids of two congruent sets are related exactly by the congruence action. As noted in \cref{subsec:RAM-limitation} however, the precision for a full rational representation is prohibitively high. 
We handle the precision issue by embedding all rational points into a finite field. More specifically, we use the \emph{modular embedding} $\phi_p:\boundedQ{U}[\ii]\to \dbF_{p}[\ii]$ for a prime $p\equiv 3\pmod 4$, defined by $\phi_p(\frac{a}{b}+\frac{c}{d}\ii) = ab^{-1}+cd^{-1}\cdot \ii$. 
Here, $\ii$ is a formal element satisfying $\ii^2=-1$. Since $x^2+1$ is irreducible over $\dbF_p$ when $p\equiv 3\pmod 4$, adjoining $\ii$ gives the essential field structure $\dbF_p[\ii]\cong \dbF_{p^2}$. For the streaming use case, modular embedding is more suitable than rounding for handling precision issues. As noted in \cref{subsec:hardness-intro}, the strong lower bound for approximate congruence testing essentially rules out any polylogarithmic space streaming algorithm based on rounding. Moreover, the prime sampling procedure is efficient in terms both space and time: for sampling a prime $p\equiv 3\bmod 8$ in the range of $[\Delta,c\Delta]$ for some constant $c$, rejection sampling requires $O(\log \Delta)$ primality tests in expectation and $O(\log \Delta)$ space, and each primality test can be done in $O(\log \Delta)$ space and $\polylog \Delta$ time, leading to a logarithmic space and time streaming-friendly subroutine.

Our use of modular embeddings follows the standard rational reconstruction paradigm \cite{Wang81,WGD82}: bounded-rational points can be faithfully embedded into finite fields by choosing primes large enough to ensure injectivity.  
For high-precision point rational reconstruction, analogous faithfulness issues have been extensively studied in computer algebra and computational algebraic geometry. A common remedy is to combine modular images over several moduli via the Chinese remainder theorem with error tolerance \cite{BDFP15,Abb17}. 
This allows reconstruction to succeed with high probability even in the presence of erroneous, or ``bad'', modular images.

A crucial distinction in our setting is that a single prime modulus must be used throughout in order to preserve the necessary multiplicative algebraic structure. To handle possible failures for a bounded number of intermediate high-precision points, we adopt a simpler probabilistic ``bad prime'' approach: sample the modulus above a suitable threshold to ensure that with high probability all relevant quantities can be reconstructed faithfully.

With this faithfulness guarantee in place, the first and third passes become straightforward. The finite-field calculations can be faithfully lifted back to $\dbQ[\ii]$, which justifies the correctness of the recentring step in the first pass. The same faithfulness also allows the standard correctness reasoning for Karp-Rabin fingerprints in the third pass to carry over.

\subsection{Product Anchors and Complex Moment}
The main challenge is to derive informative, space-efficient statistics that can be computed in as few streaming passes as possible. The algebraic representation of two-dimensional rotations leads naturally to the \emph{product-anchor} approach. 
\begin{definition}[Product Anchor] \label{def:prod-anchor}
For a multiset $S\subseteq \dbC$, its product anchor is
\(
    \Pi(S)\defeq \prod_{\substack{z\in S , z\neq 0}} z.
\)
\end{definition}
A rotation  \(\theta\) acts on each point as multiplication by \(e^{\ii\theta}\), and a product of such multiplications accumulates the rotation multiplicatively.
Another related approach is \emph{complex moment}, which was first introduced by Abu-Mostafa and Psaltis \cite{Mostafa84}. 
\begin{definition}[Complex Moments] \label{def:cpx-moment}
For a multiset $S\subseteq \dbC$ and $p,q\in\dbZ_{\geq 0}$, the \emph{order-$(p,q)$ complex moment} is defined as
\(
    M_{p,q}(S)\defeq \sum_{z\in S} z^p \overline{z}^q.
\)
We also use the shorthand $M_p(S)\defeq M_{p,0}(S)$.
\end{definition}

If $B$ is obtained from $A$ by a rotation $\rho$, then each of the $k$ nonzero points contributes one $\rho$ factor, so the ratio of the two product anchors is exactly $\rho^k$. Similarly, if $B=\rho A$, then their $k$-th moments satisfy $M_k(B) = \rho^k M_k(A)$. 
It is not a coincidence that the two methods are related. Indeed, for a set $S\subseteq \dbC$, consider its characteristic polynomial $f_S(x) = \prod_{s\in S} (x-s)$. The constant term of $f_S$ is up to sign the product anchor $\Pi(S)$, while the coefficients of $f_S$ are elementary symmetric polynomials in the elements of $S$. By Newton's identities, these can in turn be expressed in the power sums of $S$, i.e. complex moments.

While complex moments are the more established tool in the computer vision literature, the product anchor appears to be the more straightforward approach of the two. One clear advantage of product anchor is that it completely avoids the vanishing phenomenon.
When the $k$-th moments of two sets are both zero, the information on $\rho$ is annihilated. In contrast, by simply omitting the zero elements from the product, the product anchor is guaranteed to be nonzero.
On the other hand, the precision issues for the product anchor over $\boundedQ{U}^2$ (see \cref{obs:sum-prod-precision}) perhaps explain why this method has received relatively little attention.

For the complex moment method, it has long been observed in the computer vision community that it becomes ineffective when the underlying images have many moments that are identically zero or numerically close to zero. In particular, when the candidate shapes exhibit strong symmetries, a large class of commonly used moments simply vanish (\cite[Lemma 1]{FS06}).
In settings where the images are promised to be one of a few symmetric shapes, Flusser and Suk \cite{FS06} demonstrate that one can choose an appropriately tailored subset of moments to solve the geometric congruence problem. This preprocessing strategy becomes impractical for pattern classification when the shape catalogue is large, and is arguably overkill for mere congruence testing.

Our algorithms overcome the respective limitations of the product-anchor and complex-moment methods, and the underlying ideas may extend to broader applications.
For the product-anchor method, the precision issue is readily mitigated by the modular embedding technique. However, finite-field hashing introduces a new risk: recovering the rotation by solving $\widehat{x}^k=\widehat{\rho}^k$ over a finite field may produce excessive solutions, causing a prohibitive space blowup in the equality tests performed during the third pass. 
Equipped with more advanced number-theoretic tools, we show in \cref{subsec:rot-extract-PA} that the number of solutions in $\dbF_{p^2}$ can be controlled with high probability by presampling $O(\log\log n\cdot \log\frac{1}{\delta})$ moduli.

As for the complex moment approach, at the crux of our algorithm is a new structural fact of complex moments, which may be of independent mathematical and algorithmic interest. We prove that for a collection of $n$ two-dimensional equal-length nonzero rational vectors, one of its power-of-two complex moments does not vanish.
\begin{restatable}{theorem}{nonzeromoment}\label{thm:nonzero-moment}
    Let $S=\set{z_1,\ldots,z_n}\subseteq \dbQ[\ii]\setminus\set{0}$ such that $\abs{z_j}=\abs{z_k}$ for all $j,k\in [n]$.  Write $n=2^s m$ with $s\in \dbZ_{\geq 0}$ and $m$ odd, then there exists $j\in \set{0}\cup [s]$ such that $M_{2^j}(S) \neq 0$.
\end{restatable}

The special case of odd $n$ (hence $s=0$) is proven by Ball \cite{Bal73} (see \cref{thm:odd_equal_length_vectors}).
The key consequence of \cref{thm:nonzero-moment} is that a nonzero moment of a point set over $\dbQ[\ii]\setminus\set{0}$ can always be identified by checking only a logarithmic number of moments, which entails our algorithm’s efficiency with proven completeness.

\subsection{Congruence Testing in Three Dimensions} 
\label{subsubsec:TO-3D}
In addition to considering congruence testing for point sets in $\boundedQ{U}^2$, we study the same problem in $\boundedQ{U}^3$. Perhaps surprisingly, several techniques that we use for congruence testing in $\boundedQ{U}^2$ do not generalize even to a plane in $\boundedQ{U}^3$. For example, Ball~\cite{Bal73} states that no odd number of rational vectors lying on a plane in $\boundedQ{U}^2$ can sum to zero; in contrast, this statement does not hold for rational vectors lying on a plane in $\boundedQ{U}^3$, as detailed in~\cref{sec:3D}.

Rather than following a similar route to solve congruence testing for point sets in $\boundedQ{U}^3$, our techniques for the three-dimensional case rely on an interesting geometric observation. Specifically, for any two congruent $n$-point sets $A$ and $B$, given a random pair of points $p_A \in A$ and $p_B \in B$ such that $p_A$ is matched to $p_B$, we can use $\tilde{O}(n^{1/3})$ space to find, with high probability, another pair of points $q_A \in A$ and $q_B \in B$ such that $q_A$ is matched to $q_B$. The idea is to break symmetry by identifying a plane perpendicular to the line from the centroid of $A$ to $p_A$ and the corresponding plane perpendicular to the line from the centroid of $B$ to $p_B$, where each of these two planes contains $O(n^{2/3})$ points, as shown in \cref{app:3d-algo}.
The basic intuition for the three-dimensional algorithm is to obtain a pair of matched non-collinear anchor 3-tuple points $(\alpha, \beta) \in A^3 \times B^3$, and then use this pair of anchors to align the point sets for equality testing. In addition to the fingerprint and finite field techniques used in the two-dimensional case, we employ a \emph{birthday paradox} argument to sample these anchor tuples. To avoid an explosion in the number of sampled tuples, we slice the point sets into several parallel planes based on each candidate point, and prove that almost all candidate points admit a sparse plane, which successfully reduces the sample space and keeps the total number of sampled tuples sub-linear.

\section{Preliminaries} \label{sec:prel}

\subsection{Notation}
We denote $[k] \defeq \set{1,\ldots,k}$, and $\ii = \sqrt{-1}$. For $d\in \dbN$, $r>0$ and $c\in\dbC^d$, define
\[
    \calB_r(c) \defeq \set{z\in \dbC^d:\norm{z-c}=r}
    \quad\text{and}\quad
    \BB_r(c) \defeq \set{z\in \dbC^d:\norm{z-c}\leq r},
\]
and we also write $\calB_r\defeq \calB_r(0)$ and $\BB_r\defeq \BB_r(0)$. We write $\log(\cdot)$ for base-2 logarithm and $\ln(\cdot)$ for natural logarithm.

For a multiset $S$ and a function $f$, we use the shorthand $f(S)\defeq \set{f(s):s\in S}$ for the corresponding multiset. In particular, for an element $v$, we denote $S+v\defeq \set{s+v:s\in S}$ and $vS\defeq \set{vs:s\in S}$. 

As alluded to in the introduction, we customarily identify a point $(x,y)\in \dbR^2$ with a complex number $x+y\ii\in \dbC$. For the rational domains $\dbQ^2$ and $\boundedQ{U}^2$, the complex counterparts are denoted by $\dbQ[\ii]$ and $\boundedQ{U}[\ii]$ respectively. The real part and imaginary part of a complex number $z$ are denoted by $\Re(z)$ and $\Im(z)$ respectively.

\subsection{Finite-precision Rational Numbers}\label{subsec:QU}
We recall that the set of $U$-rational numbers is defined as $\boundedQ{U} \defeq \setB{\pm \frac{p}{q}:p,q\in [U]} \cup\set{0}$. We also call a complex number $z$ in $\boundedQ{U}[\ii]$ \emph{$U$-rational complex}.
As noted before, $\boundedQ{U}$ is not closed under addition and multiplication for the same precision parameter $U$. Noting that
\begin{equation*}
    \prod_{i=1}^k \frac{p_i}{q_i} = \frac{\prod_{i=1}^k p_i}{\prod_{i=1}^k q_i} \qquad \text{ and } \qquad \sum_{i=1}^k \frac{p_i}{q_i} = \frac{\sum_{i=1}^k p_i\prod_{j\neq i} q_j}{\prod_{i=1}^k q_i},
\end{equation*}
one can observe that: 
\begin{observation}\label{obs:sum-prod-precision}
    Let $r_1,\ldots,r_k\in \dbQ$ such that $r_i$ is $U_i$-rational for each $i\in [k]$.
    Then the product $\prod_{i=1}^k r_i$ is $\left(\prod^k_{i=1} U_i\right)$-rational, and the sum $\sum_{i=1}^k r_i$ is $\left(k\prod^k_{i=1} U_i\right)$-rational. 
    In particular, the product of $k$ $U$-rational numbers is $U^k$-rational, and the sum of $k$ $U$-rational numbers is $kU^k$-rational. 
    On the other hand, there are rational numbers $r_1,\ldots,r_k\in \dbQ$ where $r_i$ is $U_i$-rational for each $i\in [k]$, and the product and sum require the precisions stated in the above bounds.
\end{observation}

The above observation leads to the following corollaries for bounded precisions, which we defer to proofs to \cref{app:omitted-QU}.
\begin{corollary}\label{obs:lin-sys-precision}
    If $(x,y)$ is a solution to a $2\times 2$ linear system with $U$-rational coefficients, then $x$ and $y$ are $4U^8$-rational.
\end{corollary}
\begin{corollary}\label{obs:rational_rotation}
    Suppose $a_1,a_2,b_1,b_2\in \boundedQ{U}[\ii]$ such that $a_1\neq a_2$ and $b_1\neq b_2$. If there exists $\rho,t\in \dbC$ such that $\abs{\rho}=1$ and $b_j = \rho a_j +t$ for $j\in[2]$, then $\rho$ is $2^{10}U^{16}$-rational complex and $t$ is $2^{22}U^{35}$-rational complex. 
\end{corollary}

We also need the following classical characterization of all rotational symmetries in $\dbQ[\ii]$.
\begin{lemma}[\cite{MR80123}] 
\label{lem:niven's}
    If $z\in \dbQ[\ii]$ is a root of unity, then $z\in\set{\pm 1,\pm \ii}$.
\end{lemma}

\subsection{Karp-Rabin fingerprint} \label{app:KR-hash}
The Karp-Rabin fingerprint technique \cite{5390135} for bounded-multiplicity multisets equality serves as a fundamental algorithmic building block for this work. For convenience, we present the detailed quantitative analysis over the finite domain $[K]$ for an arbitrary $K\in \dbN$, and this extends to an arbitrary finite domain by canonical indexing.

Let $(A,B)$ be the multiset equality instance over the domain $[K]$, and $m$ be the maximum multiplicity of a point in either set. 
For a tolerance parameter $\delta\in (0,1)$, define $M=\frac{2}{\delta}\max\set{K,m}$. The fingerprint routine begins with sampling a uniform random prime $q\in [M,2M]$, and then sampling a base $b\in \dbF_q$ uniformly at random. The fingerprint for each $S\in\set{A,B}$ is given by
\[
    h_S(b) = \sum_{z\in S} b^z \in \dbF_q.
\]
It is clear that the fingerprints $h_A(b)$ and $h_B(b)$ can be maintained in the turnstile model (and hence insertion-only model). At the end of the pass, the routine declares that the multisets are equal iff $h_A(b)\pequiv[q] h_B(b)$.

It is immediate that $h_A(b)\pequiv[q] h_B(b)$ holds true regardless of the values of $b$ and $q$ whenever $A$ and $B$ are identical multisets. For the soundness, denote $f_z^S$ the multiplicity of $z\in [K]$ in $S$. Since $\abs{f_z^A - f_z^B}\leq m < K$, 
\[
    h_A(b) - h_B(b) = \sum_{z\in A\cup B} (f_z^A - f_z^B) b^z
\]
defines a non-zero $\dbF_q$-polynomial in $b$ with degree at most $K$. By Schwartz-Zippel lemma,
\[
    \Pr(h_A(b) \pequiv[q] h_B(b)) = \Pr_b(b\text{ is a root of }h_A-h_B) \leq \delta.
\]
This shows that the fingerprint routine succeeds with probability at least $1-\delta$. Finally, the space usage of the routine is $O(\log M) = O(\log K+\log m+\log \frac{1}{\delta})$.

\subsection{Number Theory}\label{subsec:finite-field}
We sometimes use the notation $\calP$ to denote the set of all primes.
For a prime $p$, we denote $\dbF_p$ the field with $p$ elements, and $\dbF_p[x]$ the polynomial ring with coefficients in $\dbF_p$ and indeterminate $x$.
For a prime $p$ with $p \equiv 3 \bmod{4}$, the quotient ring $\dbF_p[x]/\langle x^2+1\rangle$ is isomorphic to the field $\dbF_{p^2}$, while under the canonical representation we have $\dbF_p[x]/\langle x^2+1\rangle \cong \dbF_p[\ii]$. 
We use $\dbF_p[\ii]$ to emphasize the natural correspondence from $\dbZ[\ii]$ to $\dbF_p[\ii]$, to which the operators $\Re(\cdot)$ and $\Im(\cdot)$ extend naturally.

We interchangeably write $x\pequiv \lambda$ to denote $x \equiv \lambda \bmod{p}$ when $x \in \dbZ$. This notation is extended to complex elements as $x + y \ii \pequiv \lambda + \eta \ii$, meaning $x \equiv \lambda \bmod{p}$ and $y \equiv \eta \bmod{p}$ for integers $x, y$. 
We sometimes write $\Modp{\lambda}$ to highlight that $\lambda$ is viewed as an element of $\dbF_p$, and we use the analogous notation $\ModFpi{\mu}$ for elements $\mu$ of $\dbF_p[\ii]$. 

For a prime $p$, define the canonical ring homomorphism $\phi_p:\dbZ\to \dbF_p$ by $\phi_p(x) = \Modp{x}$, which extends naturally to a homomorphism from $\dbZ[\ii]$ to $\dbF_p[\ii]$. The condition $p \equiv 3 \bmod{4}$ guarantees that $\dbF_p[\ii]$ is a field. Consequently, the image of $a+b\ii\in \dbZ[\ii]$ is invertible under $\phi_p$ whenever $p\nmid a,b$, since $(a+b\ii)^{-1} = \frac{1}{a^2+b^2}(a-b\ii)$ and $a^2+b^2 \not\pequiv0$ unless $a,b\pequiv 0$.

A natural extension of this map to rational numbers (and hence $\dbQ[\ii]$) would be 
\(
    \phi_p(a/b) = \Modp{ab^{-1}}
\)
for $a \in \dbZ$ and $b \in \dbZ \setminus \set{0}$, provided that $\gcd(a,b)=1$ and $p \nmid b$. 
Whenever the inputs are well-defined, the extended map $\phi_p:\dbQ[\ii]\to \dbF_p$ preserves the homomorphism properties, namely
\begin{equation}\label{eq:homomorphism}
    \phi_p(z_1 + z_2) = \phi_p(z_1) + \phi_p(z_2) 
    \quad\text{and}\quad
    \phi_p(z_1z_2) = \phi_p(z_1) \phi_p(z_2) 
\end{equation}
for all $z_1,z_2\in \dbQ[\ii]$ well-defined under $\phi_p$. 

The modular embedding $\phi_p$ is used to compress information during streaming. To recover the $\dbQ[\ii]$ quantities from their finite-field representations, we use \emph{rational reconstruction}:
\begin{lemma}[\emph{Rational Reconstruction}, \cite{Wang81,WGD82}]\label{lem:rational_reconstruction}
    Suppose $p> 2N^2$. Then for every $x\in \dbF_p$, there is at most one pair of $(a,b)$ such that $\abs{a}\leq N$, $b\in [N]$, $\gcd(a,b)=1$ and $x\pequiv ab^{-1}$.
\end{lemma}
\begin{corollary}\label{cor:rat-recon_U-rational}
    For a given precision parameter $U$, if $p$ is a prime at least $4U^2$, then $\phi_p(z)$ is well-defined and distinct for every $z\in \boundedQ{U}[\ii]$. 
\end{corollary}
We often colloquially refer to such a $\phi_p$ map as a \emph{faithful modular embedding}.

A crucial procedure in our congruence-identification algorithms is solving equations of the form $x^k=c$ over $\dbF_{p^2}$, where $p\equiv 3 \bmod 4$ is prime, $k$ is an integer, and $c\in \dbF_{p^2}^\times\defeq \dbF_{p^2}\setminus\set{0}$. Since the multiplicative group $\dbF_{p^2}^\times$ is cyclic of order $p^2-1$, the number of solutions is controlled by the following elementary fact in algebra. 
\begin{observation}\label{obs:num-of-roots}
    Let $G$ be a cyclic multiplicative group of order $M$. For any $c\in G$, the equation $x^k=c$ has at most $\gcd(k,M)$ solutions in $G$.
\end{observation}

The \emph{prime number theorem for arithmetic progressions} states that for every modulus $c\geq 2$, the primes are asymptotically equidistributed among the coprime residue classes modulo $c$. Formally, 
\begin{theorem}[See e.g.~{\cite[Ch.~20]{Dav00}}]\label{thm:PNT-AP-asym}
    For every $c\geq 2$ and $d\in [1,c-1]$ with $\gcd(c,d)=1$,
    \[
        \pi(x;c,d) \defeq \abs{\set{p\in \calP\cap[1,x]:~p\equiv d \bmod{c}}} \sim \frac{x}{\varphi(c)\ln x}
    \]
    as $x\to \infty$.
    Here, $\varphi(\cdot)$ denotes Euler's totient function. 
\end{theorem}
We use the following simplified inequality form for small moduli derived from \cite{BMOR18}.
\begin{proposition}[See {\cite[Corollary 1.6]{BMOR18}}] \label{prop:PNT-AP}
    Let $c\in [1200]$ and $d$ be coprime with $c$. For $x\geq 50c^2$, we have
    \(
        \frac{x}{\varphi(c)\ln x}\leq \pi(x;c,d) \leq \frac{2x}{\varphi(c)\ln x}.
    \)
\end{proposition}
\section{Product-Anchor Method for Congruence Identification}\label{sec:PA-CongIden}

In this section, we present our product-anchor-based algorithm for congruence identification, thereby proving \cref{thm:CongIden-PA}. We first recall the three-stage paradigm, which is also common to the complex-moment method. In the first pass, the algorithm identifies the centroids of the two point sets for recentring. In the second pass, it computes space-efficient statistics of the recentred sets from which a small set of candidate rotations can be extracted offline. In the third pass, for each candidate rotation and its corresponding translation, the algorithm tests multiset equality under the hypothesized transformation. 

As alluded to in \cref{sec:tech-overview}, the rotation extraction utilizes the fact that if $A',B'\subseteq\dbQ[\ii]\setminus\set{0}$ are $k$-point multisets satisfying $B'=\rho A'$, then the ratio of the product of points of $B'$ to that of $A'$ is simply $\rho^k$. This relation is preserved under faithful modular embedding, so the key step is to solve the \emph{product-anchor ratio equation} $w^k = (\prod_{b\in B'} b) / (\prod_{a\in A'} a)$ in $\dbF_{p^2}$.

There are two main issues to resolve. First, the algorithm must solve the equation efficiently. If time complexity were not a concern, then the equation could be solved by brute-force search over $\dbF_{p^2}$ using $O(\log p)$ space per solution. However, the resulting $O(p^2)$ time is not remotely feasible for our prime choice regime $p=\poly(U,n,1/\delta)$.

The second issue is more acute, and it directly affects the space complexity: the number of solutions to the product-anchor ratio equation must be small. This concern cannot be avoided even if the solution set admits a compact representation, for instance by generators. Such a representation may avoid excessive space usage in the second pass, but in our paradigm the third pass requires a separate Karp-Rabin hash for each candidate rotation. Thus, the overall space complexity genuinely scales with the solution-set size.

This issue is intrinsic to the finite-field setting. Over $\dbQ[\ii]$, the torsion elements are exactly the four elements $\set{\pm 1,\pm \ii}$, whereas every element of $\dbF_{p^2}^\times$ certainly has finite order, thus the modular embedding inevitably introduces extraneous solutions for the equation $\phi_p(\rho)^k = \widehat{\lambda}$ for some $\widehat{\lambda}\in \dbF_{p^2}$. 
By \cref{lem:norm1-modular-emb,obs:num-of-roots}, the number of solutions is at most $\gcd(k,p+1)$, so it suffices to bound this quantity.
Notably, although the value of $k$ is not known until the end of the second pass, the algorithm ultimately needs to solve only one $k$-th-power equation. 
If the modulus $p$ were fixed in advance, solving
$\widehat{x}^k=\widehat{\gamma}$ in a subgroup of $\dbZ_p$ could be reduced to a
discrete-logarithm computation. In our setting however, the modulus $p$ is used for hashing purposes, so the algorithm can presample several moduli and select a suitable one a posteriori. 
Moreover, $k$ is not arbitrary but lies in $[n]$, which allows us to prove that a sampled prime $p$ satisfies $\gcd(k,p+1)\mid 4$ with probability $\Omega(1/\log\log n)$ via number-theoretic tools.

\subsection{Product-Anchor $\CongIden_{n,\boundedQ{U},2}$ Algorithm}\label{subsec:CongIden-PA}
We now present the algorithm based on product anchors. We remind readers that for a prime $p\equiv 3 \bmod{8}$, $\phi_p$ denotes the modular embedding from $\boundedQ{U}[\ii]$ to $\dbF_{p^2}$. As shown in \cref{lem:bad-primes_PA}, a suitably sampled $p$ guarantees that $\phi_p$ is a faithful modular embedding for all the relevant rational points. To facilitate reading, hat notations usually denote quantities in a finite field (e.g. $\widehat{z}$), and the corresponding no-hat notations (e.g. $z$) refer to the relevant $\dbQ[\ii]$-quantities.

\begin{enumerate}
    \setcounter{enumi}{-1}
    \item \hypertarget{PA-init}{\emph{Initialization: Prime Sampling}}
    
    Let $T = O(\log\log n\cdot \log \frac{1}{\delta})$ and $\Delta = 2^{46}n^3U^{70}\delta^{-1}$. Uniformly sample $T$ independent primes $p_1,\ldots,p_T \in [\Delta, 16\Delta]$ such that $p_j\equiv 3\bmod{8}$ for all $j\in [T]$.

    \item \emph{1st pass: Centroid Computation}

    For each $S\in \set{A,B}$ and $j\in [T]$, maintain $\phi_{p_j}(c_S) \defeq \phi_{p_j}(n)^{-1} \sum_{z\in S} \phi_{p_j}(z)$. $\phi_{p_j}(c_S)$ is the $\dbF_{p_j^2}$-embedding of the centroid of set $S$, $c_S = \frac{1}{n}\sum_{z\in S} z$.

    \item \emph{2nd pass: Product Anchors}
    \begin{itemize}
    
        \item \textcolor{darkgreen}{Streaming stage}

         
        For each $S\in \set{A,B}$ and $j\in [T]$, define the auxiliary notation $\Gamma_j(S) \defeq \set{z\in S: \phi_{p_j}(z)\neq \phi_{p_j}(c_S)}$, which corresponds to the set of non-centroid points in $S$. The algorithm maintains $k_j(S) \defeq \abs{\Gamma_j(S)}$ and
        \[
            \widehat{\Pi}_{p_j}(S) \defeq \prod_{z\in \Gamma_j(S)} (\phi_{p_j}(z) - \phi_{p_j}(c_S)).
        \]        
        $\widehat{\Pi}_{p_j}(S)$ is the $\dbF_{p_j}$-embedding of \(\Pi'(S)\defeq \Pi(S - c_S)\) (see \cref{def:prod-anchor}), that is the product of all non-centroid points in $S$ after recentring.

        \item \hypertarget{PA-2-off1}{\textcolor{darkgreen}{Offline stage (I): Rotation Recovery}}

        Choose any $j\in [T]$ such that $k_j(A) = k_j(B)$ and $\gcd(k_j(A), p_j+1)$ divides $4$. The algorithm aborts if none of the $j$'s satisfies the requirement. 
        
        Write $k \defeq k_j(A)$, $p \defeq p_j$ and $q \defeq (p+1)/4$.
        Compute $\widehat{\lambda} = \widehat{\Pi}_p(B)\cdot \widehat{\Pi}_p(A)^{-1} \in \dbF_{p^2}$. 
        Output ``not congruent'' if $\widehat{\lambda}^{p+1}\neq 1$. 
        Next, compute $e_4,e_q\in \dbZ$ such that 
        \[
            e_4+e_q=1,\quad
            \begin{cases}
                e_4 \equiv 1 \bmod 4\\ 
                e_4 \equiv 0 \bmod q
            \end{cases},\quad 
            \begin{cases}
                e_q \equiv 0 \bmod 4\\ 
                e_q \equiv 1 \bmod q
            \end{cases}
        \]
        due to Chinese Remainder Theorem.
        Compute $\ell = k^{-1} \bmod{q}$ and $\widehat{w} = \widehat{\lambda}^{\ell e_q} \in \dbF_{p^2}$. 

        \item \textcolor{darkgreen}{Offline stage (II): Rational Reconstruction}

        For each $\widehat{\rho}\in \set{\pm \widehat{w},\pm \ii \widehat{w}}$, create a candidate transformation pair $(\widehat{\rho}, \widehat{t_\rho})\defeq (\widehat{\rho}, \phi_p(c_B) - \widehat{\rho} \phi_p(c_A))\in \dbF_{p^2}$, and perform \emph{rational reconstruction} (\cref{lem:rational_reconstruction}) to recover $({\rho}, {t_\rho})\in \dbQ[\ii]$. By \cref{obs:rational_rotation}, a valid rotation $\rho$ is $2^{10}U^{16}$-rational complex and $t$ is $2^{22}U^{35}$-rational complex, so any pair violating the precision requirements can be discarded. Denote $F$ as the collection of retained candidate transformation pairs.
        
    \end{itemize}

    \item \emph{3rd pass: Equality Testing}


    Perform at most 4 parallel equality tests using Karp-Rabin hash \cite{5390135} for each of the candidate transformations $(\rho,t)$. Note that $\rho z+t$ is $2^{33}U^{52}$-rational complex for any $(\rho,t)\in F$ and $z\in \boundedQ{U}[\ii]$ (by \cref{obs:sum-prod-precision,obs:rational_rotation}). For each $(\rho,t) \in F$, we build the Karp-Rabin fingerprints for the multiset equality of $\rho A+t$ and $B$ over the domain $\boundedQ{2^{33}U^{52}}[\ii]$. The domain is enlarged so that in the turnstile model, intermediate points with precision beyond $U$ can also be inserted and later deleted.\footnote{In the insertion-only model, the fingerprints can be built over the domain $\boundedQ{U}[\ii]$. The algorithm performs the precision check $\rho a+t \in \boundedQ{U}[\ii]$ for all $a\in A$ on-the-fly, and aborts the $(\rho,t)$-subroutine if any incoming point violates the precision bound. This modification does not change the asymptotic analysis.} The algorithm outputs $(\rho,t)$ if the $(\rho,t)$-subroutine passes the fingerprint test, and outputs ``not congruent'' otherwise.   
    
    As shown in \cref{app:KR-hash}, this stage can be implemented using $O(\log U+\log n+\log \frac{1}{\delta})$ space, such that all the fingerprint routines produce the correct output with probability at least $1-\delta$. 
\end{enumerate}


\subsection{Efficient Rotation Extraction} \label{subsec:rot-extract-PA}
The \hyperlink{PA-2-off1}{offline stage (I) of the second pass} amounts to solving the equation $\widehat{x}^k=\widehat{\lambda}$ over $\dbF_{p^2}$, where $\widehat{\lambda}=\widehat{\Pi}_{p_j}(B)\cdot \widehat{\Pi}_{p_j}(A)^{-1}$, and $p_j$ is the chosen prime. If $A$ and $B$ are congruent via the rotation $\rho$ and translation $t$, then $B-c_B=\rho(A-c_A)$. Under the faithful modular embedding assumption, for each $S\in\set{A,B}$, $\widehat{\Pi}_{p_j}(S)$ is exactly the product of the nonzero points in the recentred set $S$ after applying the embedding $\phi_{p_j}$. Therefore, $\widehat{\lambda}=\phi_{p_j}(\rho)^k$, and solving $\widehat{x}^k=\widehat{\lambda}$ recovers the $\phi_{p_j}$-embedding of the rotation $\rho$. 

One preemptive check that simplifies the root extraction is to test whether $\widehat{\lambda}^{p+1}=1$. This is the ``finite-field lifting'' of the check $\abs{\rho} = 1$ in $\dbQ[\ii]$. For $s\in \dbN$, define $G_s = \set{\widehat{z}\in \dbF_{p^2}^\times: z^s = 1}$ to be the group of order-$s$ elements in $\dbF_{p^2}^\times$. We prove that norm-1 elements in $\dbQ[\ii]$ are faithfully mapped to $G_{p+1}$ in \cref{app:norm1-modular-emb}.
\begin{claim}\label{lem:norm1-modular-emb}
    Suppose $z\in \dbQ[\ii]$ is faithfully embedded in $\dbF_{p^2}$. If $\abs{z}=1$, then $\phi_p(z) \in G_{p+1}$. 
\end{claim}

The test for whether $\widehat{\lambda}\in G_{p+1}$ allows us to restrict to solving the $k$-th power equation within $G_{p+1}$, which induces additional algebraic structure that further streamlines the computation. 

The easier part is to show that, conditioned on $\gcd(k,p+1)\in \set{1,2,4}$, at most four possible values $\set{\pm \widehat{w}, \pm \ii \widehat{w}}$ can satisfy the equation.
\begin{proposition}\label{prop:eq-solve}
    Let $p$ and $k$ be as stated in the \hyperlink{PA-2-off1}{offline stage (I) of the second pass}.
    Conditioned on faithful modular mapping of $\phi_p$, all solution of the equation $\widehat{x}^k=\widehat{\lambda}$ in $G_{p+1}$ belong to the set $\widehat{R}\defeq\set{\pm \widehat{w},\pm \ii\widehat{w}}$. In particular, if $A$ and $B$ are congruent via $\rho$ and $t$, then $\phi_p(\rho)\in \widehat{R}$.
\end{proposition}
This reasoning is essentially a combination of \cref{obs:num-of-roots} and Chinese Remainder Theorem. Since $p\equiv 3\bmod 8$, $q= (p+1)/4$ is an odd number, therefore $G_{p+1}\cong G_4\times G_q$ by Chinese Remainder Theorem. 
This justifies our choice of $p\equiv 3\bmod 8$: although the weaker condition $p'\equiv 3\bmod 4$ already guarantees $\dbF_{p'}[\ii]\cong \dbF_{(p')^2}$, the stronger congruence ensures that the even-order part is generated by $G_4$, which we know $G_4=\set{\pm 1,\pm \ii}$ without requiring any additional computations. 
In $G_{q}$, since $\gcd(k,q)=1$ by the assumption $\gcd(k,p+1)\in \set{1,2,4}$, the $G_q$-part of the equation can be recovered by a suitable exponentiation. We elaborate on the details of solving such an equation in \cref{app:root-recovery}.

The more difficult part is to guarantee that $\gcd(k_j,p_j+1)$ divides $4$ for some $j\in [T]$ with high probability. Recall that $k$ denotes the number of non-centroid points in $A$ and $B$, and hence lies in the range $[1,n]$.\footnote{The case $k=0$ means that all points lie at the centroid, and can be handled trivially.} By the faithful modular embedding assumption, we can assume $k=k_j(S)$ for all $j\in [T]$ and $S\in\set{A,B}$. 
Even though $k$ is known only after the second pass, by presampling sufficiently many moduli, one of these moduli satisfies the GCD condition with high probability. The key number-theoretic ingredient is Mertens' theorem~\cite{Mer74}.
\begin{theorem}[Mertens' theorem~\cite{Mer74}]
    $\prod_{p\in \calP\cap [1,N]} (1-1/p) = \Theta(1/\log N)$.
\end{theorem}

Combining this product estimate with the prime number theorem for arithmetic progressions \cref{thm:PNT-AP-asym}, we prove that for a single $p$, the condition $\gcd(k,p+1) \mid 4$ is satisfied with probability at least $\Omega(1/\log\log n)$. Consequently, with $T=O(\log\log n\cdot \log\frac{1}{\delta})$ independently presampled primes, at least one useful modulus is sampled with probability at least $1-\delta$.
\begin{proposition}\label{prop:PA-prime-sampling}
    For any $k\in [n]$, for a uniformly sampled prime $p\in [\Delta,16\Delta]$ such that $p\equiv 3\bmod 8$,
    \(
        \Pr(\gcd(k,p+1) \mid 4) = \Omega(1/\log\log n).
    \)
\end{proposition}
We first present a proof sketch using simpler asymptotic forms of number-theoretic theorems, which illustrates the intuition for the probability bound. 
The rigorous analysis, relying on stronger estimates with explicit error bounds, will be presented in the next subsection.
\begin{proof}[Semiformal Proof of \cref{prop:PA-prime-sampling}]
    Let $E_k$ be the event of $\gcd(k,p+1) \mid 4$, and $\calF(k) \defeq \set{f\in \calP\cap [3,k]:~f\mid k}$. Given that $p\equiv 3\bmod 8$, $E_k$ is equivalent to $p\not\pequiv[f] -1$ for every $f\in \calF$. Define $M_k\defeq \prod_{f\in \calF(k)} f$. By Chinese Remainder Theorem, there are $\prod_{f\in \calF(k)} (f-1)$ mod-$8M_k$ classes where $p$ can be sampled from, of which  $\prod_{f\in \calF(k)} (f-2)$ classes satisfies the event $E_k$. 
    By the prime number theorem for arithmetic progressions, primes sampled from a sufficiently large range are nearly equidistributed among $p'-1$ mod-$p'$ classes for every prime modulus $p'$. Therefore 
    \[
        \Pr(E_k)\gtrsim \frac{\prod_{f\in \calF(k)} (f-2)}{\prod_{f\in \calF(k)} (f-1)} = \prod_{f\in \calF(k)} \left(1-\frac{1}{f-1}\right) \gtrsim C(k), \quad \text{where } C(k) \defeq \prod_{f\in \calF(k)} \left(1-\frac{1}{f}\right).
    \]
    Notice that $1-1/f$ is monotonically increasing with $f$, and $\abs{\calF(k)}\leq \log k\leq \log n$. For $n\in \dbN$, define $\beta(n)$ to be the largest prime $b$ such that $\prod_{f\in \calP\cap [3,b]} f\leq n$. Then for every $k\in [3,n]$, since $\beta(n)=\Theta(\log n)$, Mertens' theorem yields the bound 
    \[
        C(k)\geq \prod_{f\in \calP\cap [3,\beta(n)]} \left(1-\frac{1}{f}\right) =  \Theta\left(\frac{1}{\log\beta(n)}\right) =\Theta\left(\frac{1}{\log\log n}\right). \qedhere
    \]    
\end{proof}

\subsubsection{Complete proof of \cref{prop:PA-prime-sampling}}\label{app:PA-prime-sampling-rigorous}
For a rigorous lower bound analysis, we need advanced number-theoretic tools beyond the fixed-modulus asymptotic form of the prime number theorem for arithmetic progressions. 
The Chinese remainder theorem gives the correct count of admissible residue classes modulo the combined modulus $M_k$, but this residue-class count should not be interpreted as probabilistic independence among the congruence conditions modulo the different prime divisors of $k$. Consequently, bounding the prime-distribution error separately for each prime modulus does not suffice, and we need a uniform or averaged estimate that controls the accumulated error. 

For explicit error term control in estimating $\pi(x;c,a)$, we use the following concrete form of Bombieri-Vinogradov theorem \cite{Bom65,Vin65}, with simplifications from \cite{Vau05,CM06}.
\begin{theorem}[Bombieri-Vinogradov theorem]
\label{thm:BV-vaughan}
    For $Q\leq \sqrt{x}/\ln^8 x$, 
    \[
        \sum_{c\in [Q]} \max_{y\in [x]} \max_{a:\gcd(a,c)=1} \absB{\pi(y;c,a) - \frac{\Li(y)}{\varphi(c)}} = O\left(\frac{x}{\ln^2 x}\right).
    \]
    Here $\Li(x)\defeq \int_2^x \frac{dt}{\ln t}$ is the logarithmic integral.
\end{theorem}
For the remaining number-theoretic results used in this section, we omit the full details and refer the readers to standard analytic number theory texts like \cite{Dav00,IK04}. 

Recall that $M_k\defeq \prod_{f\in \calF(k)} f$, in particular $M_k\leq n$. By the inclusion-exclusion principle,
\begin{equation}\label{eq:IE}
    \Pr(E_k) = \frac{1}{\pi(16\Delta;8,3) - \pi(\Delta;8,3)} \sum_{d:d\mid M_k} \mu(d) (\pi(16\Delta;8d,a_d) -  \pi(\Delta;8d,a_d)).
\end{equation}
Here for each $d$, $a_d$ denotes the unique mod-$8d$ class which $x\equiv 3\bmod 8$ and $ x\equiv -1 \bmod d$, and $\mu(\cdot)$ denotes the M\"obius function. 

Note that $\Delta\geq n^3$, so \cref{thm:BV-vaughan} is applicable, and we obtain
\[
    \absB{\pi(\Delta;8,3) - \frac{\Li(\Delta)}{4}} \leq \max_{y\in [\Delta]} \max_{a:\gcd(a,8)=1} \absB{\pi(y;8,a) - \frac{\Li(y)}{4}} = O\left(\frac{\Delta}{\ln^2 \Delta}\right),
\]
combining with a similar calculation for $x=16\Delta$ yields
\[
    \pi(16\Delta;8,3) - \pi(\Delta;8,3) = \frac{\Li(16\Delta) - \Li(\Delta)}{4} + O\left(\frac{\Delta}{\ln^2 \Delta}\right).
\]

For $d\mid M_k$, $\gcd(8,d)=1$ and hence $\varphi(8d) = 4\varphi(d)$. By \cref{thm:BV-vaughan},
\begin{align*}
    &\sum_{d:d\mid M_k} \mu(d) (\pi(16\Delta;8d,a_d) - \pi(\Delta;8d,a_d)) \\
    &\geq \frac{\Li(16\Delta) - \Li(\Delta)}{4}\sum_{d:d\mid M_k} \frac{\mu(d)}{\varphi(d)} - \sum_{d:d\mid M_k} \absB{\pi(\Delta;8d,a_d) - \frac{\Li(\Delta)}{4\varphi(d)}}-\sum_{d:d\mid M_k} \absB{\pi(16\Delta;8d,a_d) - \frac{\Li(16\Delta)}{4\varphi(d)}}\\
    &\geq \frac{\Li(16\Delta) - \Li(\Delta)}{4} \prod_{f\in \calF(k)}\left(1+\frac{\mu(f)}{\varphi(f)}\right) 
    -2 \sum_{c\in [8n]} \max_{y\in [16\Delta]} \max_{a:\gcd(a,c)=1} \absB{\pi(y;c,a) - \frac{\Li(y)}{\varphi(c)}}\\
    &\geq \frac{\Li(16\Delta) - \Li(\Delta)}{4} \prod_{f\in \calF(k)} \left(1-\frac{1}{f-1}\right) - O\left(\frac{\Delta}{\ln^2 \Delta}\right).
\end{align*}
From the estimation $\Li(16\Delta) - \Li(\Delta) = \Theta(\Delta / \ln \Delta)$, the $O(\Delta/\ln^2 \Delta)$ terms are of lower order, therefore \cref{eq:IE} can be bounded by
\[
    \Pr(E_k) \geq \prod_{f\in \calF(k)} \left(1-\frac{1}{f-1}\right) - O\left(\frac{1}{\ln \Delta}\right).
\]
Recall that $C(k)\defeq \prod_{f\in \calF(k)} \left(1-\frac{1}{f}\right)$. Notice that
\[
    \frac{\prod_{f\in \calF(k)} \left(1-\frac{1}{f-1}\right)}{C(k)} = \prod_{f\in \calF(k)} \frac{f(f-2)}{(f-1)^2}\geq \prod_{f=3}^k \frac{f}{f-1} \cdot \prod_{f=3}^k \frac{f-2}{f-1} = \frac{k}{2}\cdot \frac{1}{k-1} \geq \frac{1}{2}.
\]
This shows that $\Pr(E_k) \geq \frac{1}{2}C(k) (1-o(1))$. As stated informally in \cref{subsec:rot-extract-PA}, the next step is to minimize $C(k)$ over $k\in [3,n]$.
\begin{claim}\label{lem:primorial-minimize} 
    Define $\beta(n)$ to be the largest prime $b$ such that $\prod_{f\in \calP\cap [3,b]} f\leq n$, and $\alpha(n)$ to be the prime product. In other words, $\alpha(n)$ is the largest primorial not exceeding $n$. Then 
    \[
        \min_{k\in \calP\in [3,n]} C(k) = C(\alpha(n)).
    \]
\end{claim}
\begin{proof}
    Note that $f\mapsto 1-1/f$ is a monotonically increasing function on $f\in [3,\infty)$. For $i\in \dbN$, denote $q_i$ the $i$-th positive odd prime (so $q_1=3$, $q_2=5$ etc.). Then if $\abs{\calF(k)}=s$, 
    \[
        C(k) = \prod_{f\in \calF(k)} \left(1-\frac{1}{f}\right) \geq \prod_{i=1}^s \left(1-\frac{1}{q_i}\right).
    \]
    The maximum value of $s$ over $k\in [3,n]$ is attained at $k=\alpha(n)$, hence the minimum value of $C(k)$ is $C(\alpha(n))$.
\end{proof}
Finally, as $\beta(n) = \Theta(\log n)$, Mertens' theorem gives $C(\alpha(n)) = \Theta(1/\log \log n)$, which leads to the desired lower bound $\Pr(E_k)=\Omega(1/\log\log n)$.

\subsection{Faithful Modular Embedding: Bad Prime Argument}\label{subsec:PA-correct}

The soundness of the algorithm in \cref{subsec:CongIden-PA} follows immediately from the standard Karp-Rabin fingerprint soundness analysis. Whenever $A$ and $B$ are not congruent, an error occurs iff one of the equality tests in the third pass incorrectly accepts. By the analysis in \cref{app:KR-hash}, this occurs with probability at most $\delta$ with the prescribed space usage for fingerprint implementation.

In the other direction of completeness, we require the faithful modular embedding assumption to hold with high probability.
Concretely, this refers to the \emph{homomorphism property} (\cref{eq:homomorphism}) and collision-free hashing for the relevant points.
From \cref{obs:rational_rotation}, all valid $U$-rational transformation pairs have precision up to $2^{22}U^{35}$ (\cref{obs:rational_rotation}). By \cref{cor:rat-recon_U-rational}, a prime choice above $2^{46}U^{70}$ ensures a faithful modular embedding for all $U$-rational points, valid rotations and translations. 

The only remaining quantities for which faithfulness of the modular embedding must be justified are the centroids. This is achieved by sampling the moduli from a sufficiently large range, which leads to our choice of the modulus threshold $\Delta$.
\begin{lemma}[Collision-free hashing for centroids] \label{lem:bad-primes_PA}
    Suppose $p_1,\ldots,p_T$ are independently and uniformly sampled primes from $[\Delta,16\Delta]$ with $p_j\equiv 3 \bmod 8$ for all $j\in [T]$.    
    Then for every input instance $(A,B)$ to $\CongIden_{n,\boundedQ{U},2}$, with probability at least $1 - \delta/2$, 
    \[
        \phi_{p_j}(c_S) \neq \phi_{p_j}(z) \qquad\text{for every $z\in S$ with $z\ne c_S$}
    \]
    for every $j\in [T]$ and $S\in \set{A,B}$.
\end{lemma}
\begin{proof}
    We show that for a single sampled prime $p$, the required property fails with probability at most $\delta/(2T)$. The lemma then follows by a union bound over the $T$ independently sampled primes.

    For each $S\in\set{A,B}$, by \cref{obs:sum-prod-precision}, $c_S$ is $nU^n$-rational complex, therefore $z - c_S$ is $2nU^{n+1}$-rational complex for every $z\in S$ with $z\neq c_S$. Write 
    \(
        z - c_S = \frac{\alpha^S_1(z)}{\beta^S_1(z)} + \frac{\alpha^S_2(z)}{\beta_2(z)}\ii,
    \)
    where $\alpha^S_1(z),\alpha^S_2(z)\in \dbZ\cap [-2nU^{n+1},2nU^{n+1}]$ and $\beta^S_1(z),\beta^S_2(z)\in [2nU^{n+1}]$. The required property $\phi_{p}(c_S) \neq \phi_{p}(z)$ is equivalent to the condition that $p$ does not divide either $\alpha^S_1(z)$ or $\alpha^S_2(z)$. Therefore, the $\phi_p$ collision-free property for centroids reduces to non-divisibility by $p$ of at most $2\cdot 2n=4n$ distinct integers $\set{\alpha^A_1(z),\alpha^A_2(z)}_{z\in A,z\neq c_A},\set{\alpha^B_1(z),\alpha^B_2(z)}_{z\in B,z\neq c_B}$ of magnitude at most $2nU^{n+1}$.

    For any $N\in \dbN$, the number of distinct prime divisors of $N$ that are at least $\Delta$ is bounded above by $\log_\Delta N$. Our choice of $\Delta$ ensures that $\Delta \geq 16n^2 T\delta^{-1}\log U$ (with ample slack), hence for a given instance $(A,B)$, the number of forbidden primes for the centroid collision-free property is at most
    \[
        4n\log_\Delta (2nU^{n+1}) \leq 8n^2\log_\Delta U \leq \frac{\delta\Delta}{2T\log\Delta}.
    \]
    By \cref{prop:PNT-AP}, the number of primes $p\in [\Delta,16\Delta]$ with $p\equiv 3\bmod 8$ is at least 
    \(
        \frac{16\Delta}{8\ln(16\Delta)} - 2\cdot\frac{\Delta}{8\ln\Delta} \geq \frac{\Delta}{\log \Delta}.
    \)
    Hence, for a uniformly sampled $p$, the required property is satisfied with probability at least $1- \frac{\delta\Delta/(2T\log\Delta)}{\Delta/\log \Delta} = 1-\delta/(2T)$.
\end{proof}

Conditioned on the modular embedding faithfulness, which holds with high probability by \cref{lem:bad-primes_PA}, all subsequent finite-field computations can be faithfully reconstructed in $\dbQ[\ii]$. 
It remains only to establish the correctness and efficiency of our second-pass rotation-extraction procedure over $\dbF_{p^2}$, which forms the theoretical basis of our product-anchor algorithm.

\subsection{Space Usage of Product-Anchor Algorithm}

At last, we justify the space usage of the product-anchor algorithm (\cref{subsec:CongIden-PA}).
One can readily check that the space usage bottleneck occurs in the second pass.

During the streaming stage, the algorithm maintains $T$ primes of magnitudes $\Theta(\Delta)$, $2T$ counters taking values in $[n]$, and $2T$ elements in $\dbF_{p^2}$ for some prime $p=\Theta(\Delta)$. Hence, the space usage is bounded by $O(T\log \Delta)$. The offline stage consists of space-efficient computations including finite-field arithmetic and exponentiation, the extended Euclidean algorithm, and rational reconstructions. These procedures require no more space than the working memory during the streaming phase. Hence, the overall space usage is 
\[
    O(T\log \Delta) = O\left(\left(\log n+\log U+\log \frac{1}{\delta}\right) \log\log n\cdot\log \frac{1}{\delta}\right).
\]

\subsection{Product-Anchor Algorithm for Points Sets}\label{subsec:PA-simple-sets}
In this section, we consider the special case of point sets, where each point appears with multiplicity one in each set. In this setting, the algorithm presented in \cref{subsec:CongIden-PA} can be further simplified, essentially because the number of non-centroid points has only two possible values: $k=n$ or $k=n-1$, depending on whether the centroid belongs to the corresponding input set. Consequently, in the \hyperlink{PA-init}{initilization phase}, the algorithm can simply rejection-sample a single prime $p$ from an appropriate range until it is ``good'' with respect to both $n$ and $n-1$, and $T$ is set to be 1 in all subsequent stages. 

More precisely, the desired properties of the prime $p$ are: $p\equiv 3 \bmod 8$, $\gcd(n,p+1)~|~4$ and $\gcd(n-1,p+1)~|~4$. 
If space efficiency is the sole concern, one can simply choose a sufficiently large range $[\Delta^*,16\Delta^*]$ containing at least one such prime. Nevertheless, a time-efficient sampling procedure can also be obtained with minor modifications to the analysis of the multiple-moduli rejection-sampling approach.
Adopting the notations in \cref{app:PA-prime-sampling-rigorous}, the efficiency of the rejection-sampling procedure amounts to lower-bounding $\Pr(E_n\wedge E_{n-1})$. Notice that $\gcd(n,n-1)=1$, therefore $E_n\wedge E_{n-1}$ is equivalent to $E_{n(n-1)}$.

With this equivalence, the success probability of the rejection sampling can be bounded using essentially the same Bombieri-Vinogradov theorem-based calculation as in \cref{app:PA-prime-sampling-rigorous} with one modification. For a prime sampled from the range $[\Delta^*,16\Delta^*]$, \cref{thm:BV-vaughan} is applicable provided that
\(
8n^2 \leq {\sqrt{\Delta^*}}/{\ln^8 \Delta^*}.
\)
It suffices to take $\Delta^*=2^{46}n^5U^{70}\delta^{-1}$, such that the success probability of each prime sampling is bounded below by $\Omega(1/\log \log n)$, and the rest of the correctness analysis in \cref{subsec:PA-correct} follows identically. 
Moreover, since the product-anchor algorithm for point sets maintains the calculations over exactly one modulus ($T=1$), the space complexity becomes $O(T\log \Delta^*) = O\left(\log n+\log U+\log \frac{1}{\delta}\right)$. This completes the justification of \cref{thm:PA-point-set}.

\section{Complex Moment Method for Congruence Identification} \label{sec:CM}

In this section, we give an alternative method for congruence identification based on \emph{complex moments}, and thus establishing \cref{thm:main}. 
At a high level, the product-anchor approach maintains one pair of statistics for the two point sets under each of several presampled moduli, such that with high probability at least one modulus permits efficient rotation recovery over the corresponding finite field. In contrast, the complex-moment approach presamples a single modulus and maintains multiple pairs of statistics chosen so that efficient rotation recovery is guaranteed. The key question is therefore how many such pairs must be maintained, which we show that $O(\log n)$ pairs of complex moments suffice in the rational setting in \cref{subsec:CM-correct}.

\subsection{Complex-Moment $\CongIden_{n,\boundedQ{U},2}$ Algorithm}\label{subsec:CM-algo}
The main difference between the two approaches is the sketch we retrieve the rotation candidate from. In the product-anchor method, the multiplicative sketch \(\Pi'(S) = \prod_{z\neq c_S}(z - c_S)\) of the recentred points is used, which is replaced by the additive sketch \(M'_k(S) = \sum_z (z - c_S)^k\).

\begin{enumerate}
    \setcounter{enumi}{-1}
    \item \emph{Initialization: Prime Sampling}
    
    Let $\Delta = 2^{86} U^{128} n^3\delta^{-1}$. Uniformly sample a prime $p\in [\Delta, 20\Delta]$ such that $p\equiv 3\pmod{16}$.

    \item \emph{1st pass: Centroid Computation}
    
    For each $S\in \set{A,B}$, maintain $\phi_{p}(c_S) \defeq \phi_{p}(n)^{-1} \sum_{z\in S} \phi_{p}(z)$. $\phi_{p}(c_S)$ is the $\dbF_{p^2}$-embedding of the centroid of set $S$, $c_S = \frac{1}{n}\sum_{z\in S} z$. We further denote \(S' \defeq S - c_S\).

    \item \emph{2nd pass: Complex Moments}
    \begin{itemize}
        \item \textcolor{darkgreen}{Streaming stage}

            For each $z\in \dbC$ and $S\in\set{A,B}$, write \(\widehat{d}_z \defeq \phi_p(\abs{z - c_S}^2 ).\) Select the first datum $z_*$ such that $\widehat{d}_{z_*}\not\pequiv 0$.
            Set $\widehat{r}^*\defeq \widehat{d}_{z_*}$.
            If no such datum exists, $A$ and $B$ consist solely of $n$ occurrences of $c_A$ and $c_B$ respectively, and this case is solved trivially.

            For $k\in\dbN$, define the $k$-th \emph{finite-field moment} of $S$ as
            \begin{equation}\label{eq:M-hat}
                \widehat{M}_k(S) \defeq 
                \sum_{z\in S: \widehat{d}_z \pequiv \widehat{r}^*} \phi_p((z-c_S)^k). 
            \end{equation}
            $\widehat{M}_k(S)$ is the $\dbF_{p^2}$-embedding of the recentred complex moment \(M'(S)\) at a chosen radius of set \(S\).
            Writing $n=2^s m$ where $s\in \dbN$ and $m$ is odd, the algorithm maintains the $2(s+1)$ power-of-two moments $\set{\widehat{M}_{2^j}(A)}_{j=0}^s$ and $\set{\widehat{M}_{2^j}(B)}_{j=0}^s$.

        \item \textcolor{darkgreen}{Offline stage (I): Rotation Recovery}

        Declare that the sets are not congruent if for some $j\in \set{0}\cup [s]$, exactly one of the $2^j$-th moments is zero.
        Otherwise, by \cref{thm:nonzero-moment}, there exists $j\in \set{0}\cup [s]$ such that both $\widehat{M}_{2^j}(A)$ and $\widehat{M}_{2^j}(B)$ is non-zero. Select the smallest such $j$, and set $\widehat{w}\defeq \widehat{M}_{2^j}(B) \cdot \widehat{M}_{2^j}(A)^{-1}\in \dbF_{p^2}$. 
        As detailed in \cref{app:power-of-two-root},  the solution set $\widehat{R}$ for $\widehat{z}^{2^j}=\widehat{w}$ can be computed efficiently, which has size at most 8.

        \item \textcolor{darkgreen}{Offline stage (II): Rational Reconstruction}
        
        For each $\widehat{g}\in \widehat{R}$, perform \emph{rational reconstructions} (\cref{lem:rational_reconstruction}) on $\widehat{g},\widehat{c}_B-\widehat{g}\widehat{c}_A\in \dbF_p[\ii]$ to obtain $\rho(\widehat{g}),t(\widehat{g})\in \dbQ[\ii]$. 
        By \cref{obs:rational_rotation}, a valid rotation $\rho$ is $2^{10}U^{16}$-rational complex and $t$ is $2^{22}U^{35}$-rational complex, so any pair violating the precision requirements can be discarded. Denote $F$ as the collection of retained candidate transformation pairs.
        
    \end{itemize}        

    \item \hypertarget{CM 3rd-pass}{} \emph{3rd pass: Equality Testing}
    
    Having recovered the candidate transformation pairs, as in the product anchor method, we run the 8 parallel Karp-Rabin equality tests (\cite{5390135}) and declare \((\rho, t)\) a valid transformation pair if the hash values agree.

\end{enumerate}

\subsection{Space Usage of Complex Moment Algorithm}
We first justify the space usage. The maximum space consumption occurs during the streaming stage of the second pass. The algorithm maintains the following quantities in the memory: $\widehat{r}^*\in \dbF_p$, $\widehat{c}_A,\widehat{c}_B\in \dbF_p[\ii]$, and $\set{\widehat{M}_{2^j}(A)}_{j=0}^s, \set{\widehat{M}_{2^j}(B)}_{j=0}^s \subseteq \dbF_p[\ii]$. The total space for storage and workspace is dominated by the $2(s+1)$ finite-field moments, hence the overall space usage of the algorithm is
\[
    O(s) \cdot O(\log p) = O\left(\log n \left(\log n+\log U+\log \frac{1}{\delta} \right)\right). 
    \qedhere
\]

\subsection{Correctness of Complex Moment Algorithm}\label{subsec:CM-correct}
For clarity, we split the correctness proof into two parts. The first part works over \(\dbQ[\ii]\) with exact arithmetic and shows that the complex moment method is correct. The second part shows that computations over the finite field \(\dbF_{p^2}\) preserve the rotation recovery property with high probability, which rests on choosing a good prime for the embedding \(\phi_p\), as we establish in \cref{lem:bad-primes}.

\subsubsection{Power-of-Two Moments as a Complete Moment-Recovery Set}\label{subsec:exact_CM_correct}

In this section, we prove that the complex moment algorithm indeed retrieves the candidate rotation set under exact arithmetic.
Translation retrieval is immediate: recentring each set at its centroid removes the translation.

The second pass of the algorithm centres on the main technical component of complex moments. 
After the first-pass preprocessing, the two-dimensional congruence problem reduces to checking whether $B'=e^{\ii\theta}A'$ for some $\theta\in [0,2\pi)$ for the recentred sets $A'$ and $B'$. This naturally motivates the use of complex moments as a tool for extracting the plausible rotations.
It is straightforward from the definition of the complex moments (\cref{def:cpx-moment}) that
\[
    M_{p,q}(S\cup T) = M_{p,q}(S) + M_{p,q}(T)
    \quad\text{ and }\quad
    M_{p,q}(S\setminus T) = M_{p,q}(S) - M_{p,q}(T),
\]
illustrating that complex moments are naturally suited to streaming implementations even in the turnstile model.
The following simple claim shows that complex moments naturally capture rotational information.
\begin{claim}\label{clm:M-rotation}
    Suppose $T=e^{\ii\theta} S$ for some $\theta\in [0,2\pi)$.     
    Then $M_{p,q}(T\cap \calB_r) = e^{\ii(p-q)\theta} M_{p,q}(S\cap \calB_r)$ for every $p,q\in \dbZ_{\geq 0}$ and $r>0$.
\end{claim}
\begin{proof}
    Recall that \(\calB_r = \set{z\in \dbC:\abs{z}=r}\). Since $e^{\ii\theta}$ defines an isometry, we have $T\cap \calB_r = \rho S\cap \calB_r = \rho(S\cap \calB_r)$. Thus
    \[
        M_{p,q}(T \cap \calB_r) = \sum_{w\in S\cap \calB_r} (e^{\ii\theta} w)^p \overline{(e^{\ii\theta} w)}^q = e^{\ii(p-q)\theta} \sum_{w\in S\cap \calB_r} w^p \overline{w}^q = e^{\ii(p-q)\theta} M_{p,q}(S\cap \calB_r). \qedhere
    \]
\end{proof}
Moreover, $M_{p,q}(S) = \overline{M_{q,p}(S)}$, and $M_{p,q}(S \cap \calB_r) = r^{2q} M_{p-q,0}(S \cap \calB_r)$ for a fixed $r>0$. 
Therefore to study rotations at a fixed radius, it suffices to consider the moments $M_p(S)= M_{p,0}(S)$. In this notation, the quantities maintained in the algorithm, $M'_{2^j}(A)$ and $M'_{2^j}(B)$ defined in \cref{eq:M-hat}, are indeed the $2^j$-th recentred moments $M_{2^j}(A'\cap \calB_r)$ and $M_{2^j}(B'\cap \calB_r)$.

\cref{clm:M-rotation} provides a viable approach to recover $e^{\ii\theta}$: provided that $B'=e^{\ii\theta} A'$, and for some $j$ and $r>0$ the moments $M_{2^j}(A'\cap \calB_r)$ and $M_{2^j}(B'\cap \calB_r)$ are non-zero, we have
\begin{equation}\label{eq:moment-quotient}
    \frac{M_{2^j}(B'\cap \calB_r)}{M_{2^j}(A'\cap \calB_r)} = e^{\ii 2^j\theta}.
\end{equation}
On the other hand, if for some $j$ exactly one of $M_{2^j}(A'\cap \calB_r)$ and $M_{2^j}(B'\cap \calB_r)$ is zero, their moduli are unequal and thus $B'\neq e^{i\theta} A'$ for any $\theta$, hence we can conclude that $A$ and $B$ are not congruent for certain. 
While the case when both $2^j$-th moments vanish leaves the outcome inconclusive, one of our main technical contributions (\cref{thm:nonzero-moment}), restated below, asserts that for a set of $n$ equal-magnitude rational vectors, the power-of-two moments cannot all vanish simultaneously, hence the algorithm can extract $(e^{\ii \theta})^{2^j}$ for some $j$ or correctly abort on negative instances.

The remainder of this section is devoted to proving \cref{thm:nonzero-moment}. 
We recap the theorem for readers' convenience. 
\nonzeromoment*
The base case is a classical result of Ball \cite{Bal73} on the constructibility of equilateral polygons on the integer lattice, restated here in terms of complex moments.

\begin{theorem}[\cite{Bal73}] \label{thm:odd_equal_length_vectors}
    Let $S=\set{z_1,\ldots,z_n}\subseteq \dbQ[\ii]$ such that $\abs{z_j}=\abs{z_k}$ for any $j,k\in [n]$. If $n$ is odd, then $M_{1}(S)=\sum_{j=1}^n z_j$ is non-zero.
\end{theorem}
For a set $T=\set{u_1,\ldots,u_N}\subseteq \dbQ[\ii]$ and $p\in\dbN$, define the auxiliary set operations $T^{(p)}\defeq \set{u_j^p:j\in [N]}$ and $T_\times\defeq \set{u_ju_k:1\leq j<k\leq N}$. Note that if all elements in $T$ have the same magnitude. then so do all elements in $T^{(p)}$ and $T_\times$.

For any $p\in \dbN$, the $p$-th moments of the sets $T$, $T^{(2)}$ and $T_\times$ satisfy 
\begin{equation*}
    M_p(T)^2 = \left(\sum_{j=1}^N u_j^p\right)^2 = \sum_{j=1}^N u_j^{2p} +2\sum_{j<k} u_j^p u_k^p = M_p(T^{(2)}) + 2M_p(T_\times),
\end{equation*}
thus
\begin{equation}
    M_p(T_\times) = \frac{1}{2}(M_p(T)^2 - M_p(T^{(2)})) = \frac{1}{2}(M_p(T)^2 - M_{2p}(T)). \label{eq:M-cross}
\end{equation}

\begin{proof}[Proof of \cref{thm:nonzero-moment}]
    Suppose the contrary that $M_{2^j}(S) = 0$ for all $j\in\set{0}\cup [s]$. 
    For $N\in \dbN$, define $h(N)$ the order of 2 dividing $N$, i.e. the largest integer $s'$ such that $2^{s'}$ divides $N$. In this notation, $h(\ell)=0$ for any odd number, and $h(n)=s$. We define a sequence of $s+1$ multisets in which the last multiset contains an odd number of elements, and the moment of each multiset is related to the power-of-two moments of $S$.

    Define $T_0 = S$ and $T_j = (T_{j-1})_\times$ for $j\in [s]$. Note that $\abs{T_0} = n = 2^s m$ where $m$ is odd, and $\abs{T_j} = \binom{\abs{T_{j-1}}}{2}$ for $j\in [s]$, thus $h(\abs{T_j}) = s - j$ for every $0\leq j\leq s$. Now as $h(\abs{T_s}) = 0$ and every element in $T_s$ has the same magnitude, \cref{thm:odd_equal_length_vectors} implies that $M_1(T_s)\neq 0$.

    In \cref{app:M(T_j)}, we prove the following lemma by iteratively applying \cref{eq:M-cross}:
    \begin{lemma}\label{lem:M_1(T_s)}
        $M_1(T_s)$ is a polynomial in $\set{M_{2^j}(S)}_{j=0}^s$ with zero constant term. 
    \end{lemma}  
    As $M_{2^j}(S)=0$ for all $j\in \set{0}\cup [s]$, $M_1(T_s) = 0$ by this lemma, which is a contradiction.
\end{proof}

Lastly, the equality testing in the third pass can be analyzed identically to that for the product-anchor algorithm. Its correctness follows from \cref{app:KR-hash}.

\subsubsection{Faithfulness of the finite-field computation}
Having established correctness over $\dbQ[\ii]$ with exact arithmetic, we now show that carrying the computation into \(\dbF_{p^2}\) preserves it. The correctness hinges on the choice of a suitable prime $p$ for the modulo map $\phi_p$. 
We formalize the requirements for selecting a \emph{good prime} below and establish that a uniformly chosen prime from a modest range is good with high probability.

As noted in \cref{subsec:finite-field}, $p\equiv 3\bmod 4$ is necessary to guarantee that the field extension $\dbF_p[\ii]$ is defined as intended. In the rational complex case, Niven's theorem (\cref{lem:niven's}) guarantees that $z^{2^j}=w$ in $\dbQ[\ii]$ has at most 4 rational solutions.
To ensure that $\abs{\widehat{R}}$ remains constant in the finite field setting, we strengthen the requirement on $p$ to $p\equiv 3\bmod 16$, so that $\abs{\widehat{R}}\leq \max_{j\in \dbN}\set{\gcd(2^j,p^2-1)} = 8$. The other aspects of the algorithm are not affected by this change.

The next lemma defines the required bad-prime conditions tailored for the complex moment algorithm, which can be viewed as a generalization of \cref{lem:bad-primes_PA} for the product-anchor algorithm.
\begin{lemma}\label{lem:bad-primes}
    Let $\Delta = 2^{86} U^{128} n^3\delta^{-1}$.
    For every input instance $(A,B)$ to $\CongIden_{n,\boundedQ{U},2}$, with probability at least $1 - \delta/2$, a uniform random prime $p\in [\Delta,20\Delta]$ with $p\equiv 3 \bmod 16$ satisfies the following properties:
    \begin{enumerate}[label=\textup{(\alph*)}]
        \item\label{prop.distinct}
        All $z\in \boundedQ{2^{22}U^{35}}[\ii]$ are well-defined under $\phi_p$ and their images under $\phi_p$ are distinct;
        \item\label{prop.a}  
        $\widehat{c}_A \pequiv \ModFpi{\phi_p(z)}$ for each $z\in A$ with $z\neq c_A$, similarly for $B$;
        \item\label{prop.b} 
        The distinct values in the set $\Psi\defeq \set{\abs{z-c_A}^2:z\in A} \cup \set{\abs{z-c_B}^2:z\in B}$ are mapped to distinct images under $\phi_p$;
        \item\label{prop.c}
        For $j\in \set{0}\cup[s]$, $M'_{2^j}(A) = 0$ (defined in \cref{eq:M-hat}) iff $\widehat{M}_{2^j}(A) \pequiv \ModFpi{0}$, similarly for $B$;
        \item\label{prop.d}
        If $j\in \set{0}\cup [s]$ is the minimum value such that $M'_{2^j}(A),M'_{2^j}(B)\neq 0$, for any two distinct $\rho,\rho' \in \boundedQ{2^{22}U^{35}}[\ii]$ in which $\rho^{2^j}\neq (\rho')^{2^j}$, then $\phi_p(\rho^{2^j})\not\pequiv \phi_p((\rho')^{2^j})$.
    \end{enumerate}
\end{lemma}

From \cref{lem:rational_reconstruction}, the lower bound on the prime-selection range along with \cref{cor:rat-recon_U-rational} guarantees that Property~\ref{prop.distinct} always holds. Consequently, distinct points in $\boundedQ{U}[\ii]$, distinct rotations and translations in $\boundedQ{2^{22}U^{35}}[\ii]$ both admit well-defined and collision-free hashes.

For the remaining properties, Property \ref{prop.a} ensures that at least one datum $z$ satisfies $\widehat{d}_z\neq 0$ except in the degenerate case where both input sets comprise a single repeated point; Property \ref{prop.b} guarantees that the points collected for finite-field moment evaluations indeed correspond to points equidistant from the true centroids; Property \ref{prop.c} preserves, under reduction modulo $p$, the existence of a non-vanishing power-of-two moment established in the real-register setting (\cref{thm:nonzero-moment}); Property \ref{prop.d} eliminates the need of considering rotations other than those values of $\rho$ satisfying $\rho^{2^j} = M'_{2^j}(B)/M'_{2^j}(A)$.

\begin{proof}[Proof of \cref{lem:bad-primes}]
    As mentioned above, the range of prime selection ensures that Property \ref{prop.distinct} holds, the $\dbF_p[\ii]$-hashes of the set $\set{\phi_p(z): z \in \boundedQ{U}[\ii]}$ are all distinct. Also, $\Modp{n^{-1}}$ is well-defined, therefore $\widehat{c}_A$ and $\widehat{c}_B$ are well-defined under $\phi_p$.

    Properties \ref{prop.a} -- \ref{prop.d} can each be recast as a requirement that the numerators of the lowest-term rational representations of certain bounded-precision rational numbers are not divisible by $p$.

    Define $S_{(b)}$ to be the set of non-zero integers in which non-divisibility by $p$ for each $z\in S_{(b)}$ implies Property \ref{prop.a}. Let $N_{(b)}=\abs{S_{(b)}}$, and $U_{(b)}=\max\set{\abs{z}:z\in S_{(b)}}$. Define the similar quantities for Properties \ref{prop.b} -- \ref{prop.d}.

    For Property \ref{prop.a}, we give the justification for $A$ and the same argument works for $B$. Note that $c_A$ is $nU^n$-rational complex, therefore $z - c_A$ is $2nU^{n+1}$-rational complex for each $z\in A\setminus\set{c_A}$ by \cref{obs:sum-prod-precision}. One can write
    \(
        z - c_A = \frac{\alpha_1(z)}{\beta_1(z)} + \frac{\alpha_2(z)}{\beta_2(z)}\ii,
    \)
    where $\alpha_1(z),\alpha_2(z)\in \boundedQ{2nU^{n+1}}$ and $\beta_1(z),\beta_2(z)\in \boundedQ{2nU^{n+1}}\setminus\set{0}$. 
    The required property $\ModFpi{\phi_p(z) - \widehat{c}_A}\not\pequiv 0$ is equivalent to the condition that $p$ does not divide both $\alpha_1(z)$ and $\alpha_2(z)$. 
    Property \ref{prop.a} reduces to non-divisibility by $p$ of at most $n$ pairs of integers $(\alpha_1(z),\alpha_2(z))$ in $[2nU^{n+1}]^2$. Therefore $N_{(b)} \leq 4n$ and $U_{(b)} \leq 2nU^{n+1}$.

    In \cref{app:2d-bad-primes}, we present the calculations for the following bounds:
    \begin{claim}\label{clm:bad-primes-count}
        $N_{(c)}\leq n(2n-1)$, $U_{(c)}\leq 2^{11}n^8 U^{8n+8}$; $N_{(d)}\leq 4\log n+4$, $U_{(d)}\leq n^{n+1}2^{n^3}U^{n^2}$; $N_{(e)}\leq 2^{80}U^{128}$, $U_{(e)}\leq 2n^2 2^{22n^2}U^{32n}$.
    \end{claim}

    For $N\in \dbN$, with our choice of \(\Delta\), a rough upper bound for the number of distinct prime factors of $N$ is $\log_\Delta N$. Let $K$ be the total number of distinct prime factors of all integers involved in the non-divisibility conditions in Properties \ref{prop.a} -- \ref{prop.d}, so all five required properties are satisfied unless $p$ is one of these $K$ primes. Note that 
    \begin{align*}
        K &\leq N_{(b)}\log_\Delta U_{(b)} + N_{(c)}\log_\Delta U_{(c)} + N_{(d)}\log_\Delta U_{(d)} + N_{(e)}\log_\Delta U_{(e)},
    \end{align*}
    which is at most $K_*\defeq 2^{85} U^{128} n^3 / \log \Delta$ after simplifications.     
    As implied by the prime number theorem for arithmetic progressions (see \cite[Corollary 1.6]{BMOR18}), the number of primes $p$ in $[\Delta,20\Delta]$ with $P\equiv 3\bmod 16$ is at least
    
    \[
        \frac{20\Delta}{8\log(20\Delta)} - \frac{\Delta}{4\log\Delta} \ge \frac{20\Delta}{16\log\Delta} - \frac{\Delta}{4\log\Delta} = \frac{\Delta}{\log\Delta} = \frac{2K_*}{\delta}.
    \]
    Therefore, for a uniform random prime $p$ in $[\Delta,20\Delta]$ with $P\equiv 3\bmod 16$, all five properties are satisfied with probability at least $1 - \delta/2$.
\end{proof}

\begin{proof}[Proof of \cref{thm:main}]
    The algorithm in \cref{subsec:CM-algo} errs either when the prime $p$ chosen at the initialization does not satisfy the five properties in \cref{lem:bad-primes}, or the equality testing phase in the \hyperlink{CM 3rd-pass}{third pass} produces a false positive. The former case occurs with probability $\delta/2$ as shown in \cref{lem:bad-primes}. The latter case follows directly from the product anchor section.
    By applying a union bound over the failures of the five properties in \cref{lem:bad-primes} together with the failure probabilities of the (at most) eight equality tests in the third pass (see \cref{app:KR-hash}), we conclude that the algorithm succeeds with probability $1-\delta$.
\end{proof}

\section{Congruence Signature}\label{sec:geometric-hashing}

In this section, we introduce yet another method for $\CongIden_{n,\boundedQ{U},2}$ based on the product anchor (\cref{subsec:CongIden-PA}) method. With just one extra pass, we reduce the space complexity by one \(\log\log n\) factor, and most importantly, the new algorithm builds a signature for each set in the stream, after which any pairwise congruence query is answered offline, without revisiting the data.

We consider an auxiliary problem called the \emph{congruence signature} problem, $\CongSign_{n,\boundedQ{U},2}$, defined as follows. For a point multiset $S\subseteq \boundedQ{U}^2$ of $n$ points, a congruence signature algorithm output a signature $\tilde{H}(S)\in\set{0,1}^s$ and an ordered anchor point pair $(a,b)\subseteq \boundedQ{U}^2\times \boundedQ{U}^2$, and we refer to $s$ as the \emph{signature length}.
We show that a signature with a good collision-free guarantee can be obtained with high probability.
\begin{lemma}\label{lem:geom-signature}
    For each $m\in \dbN$, there exists a 4-pass $O(\log n+\log U+\log m + \log \frac{1}{\delta})$-space randomized streaming algorithm for $\CongSign_{n,\boundedQ{U},2}$ of signature length \(O(\log U + \log m+ \log \frac{1}{\delta})\), such that with probability at least $1-\delta/m^2$:
    \begin{enumerate}[label=(\alph*)]
        \item Any two non-congruent sets receive distinct signatures; 
        \item For any two congruent sets, a valid rotation and translation can be computed using their anchor sets.
    \end{enumerate}
\end{lemma}

With a congruence signature algorithm $\tilde{H}$ in place, we construct a congruence hash function from the individual signatures: for query sets $A_1,\ldots,A_m\subseteq \boundedQ{U}^2$, define the congruence hash function $H:[m]\to\set{0,1}^s$ by $H(i) = \tilde{H}(A_i)$ for each $i\in [m]$.
By a union bound on the $\binom{m}{2}$ pairs of query sets, \cref{lem:geom-signature} implies that the congruence hash function is valid with probability $1-\delta$. 
Since the congruence signatures of the individual query sets can be computed in parallel, we obtain a $4$-pass $\CongHash^m_{n,\boundedQ{Q},2}$ streaming algorithm whose space usage is $m$ times that of a single instance of $\CongSign_{n,\boundedQ{Q},2}$, yielding the desired algorithm as claimed in \cref{thm:geomhash-database}.

We focus on the streaming complexity of the congruence signature problem for the remainder of this section.
Our $\CongSign_{n,\boundedQ{U},2}^m$ algorithm follows the same principle as the product-anchor algorithm (\cref{subsec:CongIden-PA}) or complex-moment algorithm (\cref{subsec:CM-algo}) to transform each query set into a canonical form.
Since congruence queries are answered offline, each set must be assigned its canonical frame without knowing the other query sets. We utilize ideas from the product-anchor method for choosing anchors, and the complex-moment method for choosing canonical radii, which are space-efficient statistics that can be determined for each set separately. 
Accordingly, every set signature generated by our algorithm depends only on information from the set itself, together with a few random primes for hashing purposes.
In this way, we circumvent the difficulties in adapting standard RAM-model algorithmic techniques such as computing the lexicographically smallest rotation position. The approach of finite-field embedding of streaming-friendly statistics allows us to obtain a succinct representation alternatively.

Unlike the previous two methods, the congruence signature algorithm does not directly extend to the turnstile model since this algorithm relies on the exact recovery of the minimum element in the stream. As we will show in \cref{app:min_element}, this task requires polynomial space up to logarithmic passes when deletion is present. It remains an interesting open question of whether the backbone technique of modular embedding can be combined with other methods to yield an efficient turnstile congruence hashing algorithm.

\subsection{Congruence Signature Algorithm}\label{subsec:GeomSign}
The congruence signature algorithm is as follows:
\begin{enumerate}
    \setcounter{enumi}{-1}
    \item \emph{Initialization: Random Selection of Hashing Function}
    
    Let \(\Delta' = 2^{87}U^{128}n^6m^3\delta^{-1}\). Uniformly sample a prime \(p\in[\Delta',20\Delta']\) with \(p\equiv 3 \bmod{4}\) . Adopting a similar analysis, in \cref{lem:bad-primes-geomhash} we show that, with probability at least \(1-\delta/(2m^2)\), all calculations involving \(\phi_p\) on \(\boundedQ{U}[\ii]\) remain well-defined. All subsequent finite-field operations can therefore be treated as faithful.

    \item \emph{1st pass: Centroid Computation}
    
    For the centroid $c_S$ of the set $S$, compute $\widehat{c}_S = \phi_p(c_S) \in \dbF_p[\ii]$.
    
    \item \hypertarget{GeomSign 2nd pass}{\emph{2nd pass: Canonical Radius Computations}}
    
    Following the same approach in the second pass of \Cref{subsec:CongIden-PA}, compute the product anchor \(\widehat{\Pi}_p(S)\) along with the number of non-centroid points \(k\).
    
    Simultaneously, we find a canonical radius. A natural choice is the distance from the centroid to the closest point in $S$, which is not feasible due to potential precision issues. 
    Instead, we take the ``$\dbF_p$-minimum''. Concretely, consider the mod-$p$ distance set
    \[
        \widehat{\Psi}(S) \defeq \set{\phi_p\left( \abs{z-c_S}^2\right):z\in S} \subseteq \dbF_p.
    \]
    Viewing $\dbF_p\setminus\set{0}$ as $\set{1,2,\ldots,p-1}\subseteq \dbZ$, we define $\widehat{r}$ and $\widehat{r}'$ as the minimum value and the second smallest value in $\widehat{\Psi}(S)\setminus\set{0}$ respectively. We remind readers that $\widehat{r}$ and $\widehat{r}'$ are in $\dbF_p$.

    

    \item \hypertarget{GeomSign 3rd pass}{\emph{3rd pass: Canonical $k$-th Radius and Reference Points Collection}}
    
    

        Define the mod-$p$ product anchor distance set by
        \[
            \widehat{\Psi}_k(S) \defeq \setB{\phi_p\left(\abs{(z-c_S)^k - \widehat{\Pi}'_p(S)}^2\right):z\in S} \subseteq \dbF_p.
        \]
        In this pass, for each point \(z\in S\) satisfying
        \[
            \phi_p(\abs{z-c_S}^2)\in\{\widehat{r},\widehat{r}'\},
        \]
        we evaluate its corresponding value in \(\widehat{\Psi}_k(S)\). Among these values, we maintain the two smallest distinct values
        \(\widehat{r}_k<\widehat{r}'_k\), together with a buffer \(A\) containing all points attaining either value.
        The resulting set \(A\) satisfies both
        \begin{equation} \label{eq:anchor-condition}
            \phi_p(\abs{z-c_S}^2)\in\{\widehat{r},\widehat{r}'\} \qquad
            \text{and}
            \qquad \phi_p(\abs{(z-c_S)^k-\widehat{\Pi}'_p(S)}^2)\in\{\widehat{r}_k,\widehat{r}'_k\}.
        \end{equation}
        We refer to the points remaining in \(A\) as \emph{reference points}. In \cref{clm:anchor-size}, we prove that \(A\) is of constant size.


        
    
        
    \item \hypertarget{GeomSign 4th pass} {\emph{4th pass: Signature Generation}}
    
        If \(\lvert A\rvert=1\) and assuming that all computations are faithful, \(S\) consists of \(n\) occurrences of the same point, and therefore can be solved trivially.

    In the other case, there are at least two reference points in $A$.
    For each reference pair \(\{a,b\}, a \neq b\), we map each point \(z \in S\) to a tuple consisting of its distances to \(a\) and \(b\), together with an indicator for the relative position of $z$ from the line joining \(a\) and \(b\). 
    Concretely, for each $z\in S$, let \(\sigma_{a,b}(z) \in \set{-1,0,+1}\) be the sign of \(\det(b-a,\,z-a)\)\footnote{Here for the determinant computation, we treat $a,b,z$ as two-dimensional vectors.}, and
    \[
        D_{a,b}(z) \defeq \set{\abs{z-a}^2, \abs{z-b}^2} \subseteq \boundedQ{2^{11}U^{16}}^2.
    \]
    Here, the precision bound follows from \cref{obs:sum-prod-precision}. Define the size-$n$ multiset
    \[
        \Gamma_{a,b}(S) \defeq \set{(D_{a,b}(z), \sigma_{a,b}(z)):z\in S} \subseteq \left(\boundedQ{2^{11}U^{16}}^2\times \set{-1,0,+1} \right)^n.
    \]

    At the beginning of this pass, uniformly sample a random prime $[2^{47}U^{64}m^2\delta^{-1},2^{48}U^{64}m^2\delta^{-1}]$ for all parallel Karp-Rabin fingerprint runs in this pass. For each reference pair, $\set{a,b}$, maintain the Karp-Rabin fingerprint $F_{a,b}(S)\in \dbF_q$ of $\Gamma_{a,b}(S)$.
    

    In the offline phase, select $\tilde{H}(S)$ to be the lexicographically smallest tuple
    \[
        (\abs{a-b}^2, F_{a,b}(S) ) \in \boundedQ{2^{11}U^{16}}^2\times \dbF_q
    \]
    among all reference pairs $\set{a,b}$.
    We also return this ordered reference pair \((a,b)\) as the anchor set. When two sets share a signature, we use their anchors \((a,b)\) and \((a',b')\) to solve for the rotation \(r\) and translation \(t\).
\end{enumerate}

\subsection{Correctness and Space usage of the \(\CongHash\) Algorithm}
We slightly modify the bad-prime argument of Lemma~\ref{lem:bad-primes} for our \(\CongHash\) algorithm.
\begin{lemma}\label{lem:bad-primes-geomhash}
    Let \(\Delta' = 2^{87}U^{128}n^6m^3\delta^{-1}\). Let $\mathcal{F} = (S_1,\ldots,S_m)$ be an input  to $\CongHash^m_{n,\boundedQ{U},2}$. With probability at least \(1 - \delta/(2m^2)\), for every $S\in \mathcal{F}$, a uniformly random prime \(p\in[\Delta',20\Delta']\) with \(p\equiv 3 \bmod 4\) simultaneously satisfies the following properties:
    \begin{enumerate}[label=\textup{(\alph*')}]
        \item\label{geomhash.prop.a}
        All $z\in \boundedQ{2^{22}U^{35}}[\ii]$ are well-defined under $\phi_p$ and their images under $\phi_p$ are distinct;

        \item\label{geomhash.prop.b}  
        For every point \(z\in S\setminus\set{c_S}\), we have
        \(
        \widehat{c}_S \neq {\phi_p(z)} .
        \)

        \item\label{geomhash.prop.c} 
        Let \(\Psi(S) \defeq \set{\abs{z-c_S}^2 : z\in S}\subseteq \dbQ\). The distinct values in \(\Psi(S)\) are mapped to distinct images under \(\phi_p\).

        \item \label{geomhash.prop.d}
        Let \(k\) be the number of non-centroid point of \(S\), for every \(z\in S\), if \((z - c_S)^k \neq \Pi'(S)\) then
        \(
        \phi_p\left((z - c_S)^k\right) \neq \phi_p\left(\widehat{\Pi}_p(S)\right) .
        \)

        \item\label{geomhash.prop.e} 
        Let \(k\) be the number of non-centroid point of \(S\), and
        \(
        \Psi_k(S) \defeq \setB{\lvert(z -c_S)^k - \widehat{\Pi}_p(S)\rvert^2 : z\in S}\subseteq \dbQ.
        \)
        The distinct values in \(\Psi_k(S)\) are mapped to distinct images under \(\phi_p\).

        \item\label{geomhash.prop.f}
        Let \(k\) be the number of non-centroid point of \(S\), for any two distinct $\rho,\rho' \in \boundedQ{2^{22}U^{35}}[\ii]$ in which $\rho^k\neq (\rho')^k$, then $\phi_p(\rho^k)\not\pequiv \phi_p((\rho')^k)$.
        
    \end{enumerate}
\end{lemma}

Properties \ref{geomhash.prop.a}, \ref{geomhash.prop.b} and \ref{geomhash.prop.c} serve for the same reason as in the previous section. Properties~\ref{geomhash.prop.d} and \ref{geomhash.prop.e} are stronger variants tailored to \(\CongHash\): Property~\ref{geomhash.prop.d} guarantees that there exists at least one datum \(z\in S\) with non-zero \(\phi_p((z-c_S)^k)\) except for the degenerate case, and Property~\ref{geomhash.prop.e} guarantees that different product anchor distances of $S$ remain distinct under \(\phi_p\), so that our algorithm collects exactly the points at the intended distance. Property \ref{geomhash.prop.f} is for bounding the number of solutions.

\begin{proof}[Proof of \cref{lem:bad-primes-geomhash}]
Property \ref{geomhash.prop.a} holds by the same reasoning for Property \ref{prop.distinct} in \cref{lem:bad-primes}. The argument follows the proof of \cref{lem:bad-primes}: each property can be reduced to requiring that \(p\) avoids dividing the numerators of certain lowest-term rationals of bounded precision.

Fix an arbitrary \(S\in \mathcal{F}\). We define $S_{(f')}$ as the set of non-zero integers such that if $p$ divides none of the integers in $S_{(f')}$, then Property \ref{geomhash.prop.f} holds. Let \(N_{(f')}\defeq |S_{(f')}|\) and \(U_{(f')} = \max\{|t|:t\in S_{(f')}\}\). Analogous quantities are defined for Properties \ref{geomhash.prop.b}--\ref{geomhash.prop.e}.

Properties \ref{geomhash.prop.b}, \ref{geomhash.prop.c}, and \ref{geomhash.prop.f} are handled identically to Properties \ref{prop.a}, \ref{prop.b}, and \ref{prop.d} in \cref{lem:bad-primes,clm:bad-primes-count}. We will only present the bounds for the quantities relevant to Properties \ref{geomhash.prop.d} and \ref{geomhash.prop.e} in \cref{app:geomhash-bad-primes}:
\begin{claim}\label{clm:geomhash-bad-primes-count}
    \(N_{(d')}\le 2n\),
    \(
        U_{(d')}\le 2^{3n^2}n^{4n^2}U^{4n^3}
    \);
    \(N_{(e')}\le n^2\),
    \(
        U_{(e')}\le 2^{27n^2}n^{32n^2}U^{32n^3}.
    \)
\end{claim}

For a fixed set $S$, let \(K_S\) be the number of distinct primes that divide at least one integer appearing in the non-divisibility conditions for Properties \ref{geomhash.prop.b}--\ref{geomhash.prop.f} for the set $S$. As before,
\begin{align*}
        K_S &\leq N_{(b')}\log_{\Delta'} U_{(b')} + N_{(c')}\log_{\Delta'} U_{(c')} + N_{(d')}\log_{\Delta'} U_{(d')} + N_{(e')}\log_{\Delta'} U_{(e')} + N_{(f')}\log_{\Delta'} U_{(f')}.
    \end{align*}
Using the same bounds as in \cref{clm:bad-primes-count} for Properties \ref{geomhash.prop.b}, \ref{geomhash.prop.c} and \ref{geomhash.prop.f}, their total contribution is at most \(2^{85}U^{128}n^3/\log \Delta'\). For the two other terms, \cref{clm:geomhash-bad-primes-count} yields the (very) crude upper bound 
\[
N_{(d')}\log_{\Delta'} U_{(d')} + N_{(e')}\log_{\Delta'} U_{(e')} \le \frac{2^7 n^6 U^2}{\log \Delta'}.
\]

Therefore $K'\defeq \sum_{S\in \mathcal{F}} K_S\leq \sum_{S\in \mathcal{F}} 2^{86} n^6 U^{128} / \log \Delta' = 2^{86} n^6 U^{128} m / \log \Delta'$.

Similar to the proof of \cref{lem:bad-primes}, by the same prime-counting argument, the interval \([\Delta',20\Delta']\) contains at least \(2K'm^2\delta^{-1}\) primes where \(p\equiv 3 \bmod 4\). Therefore, for a uniformly random such prime \(p\), all six properties hold with probability at least \(1-K'/(2K'm^2\delta^{-1}) = 1-\delta/(2m^2)\).
\end{proof}

With the bad-prime argument in place, we are in a position to prove the correctness of the algorithm. The part before product anchor computation is identical to that in the previous section, so we omit it here and focus on the additional hashing steps. In particular, we prove two properties: 
\begin{itemize}
    \item[(i)] \emph{Completeness}: For any two congruent point sets \(A,B\in \mathcal{F}\), the \(\CongSign\) algorithm (see \cref{subsec:GeomSign}) outputs the same hash value on \(A\) and \(B\), and we can identify the rotation \(r\) and translation \(t\) that map \(A\) onto \(B\).
    \item[(ii)] \emph{Soundness}: For any two non-congruent point sets \(A,B\in \mathcal{F}\), the \(\CongSign\) algorithm outputs colliding hash values on \(A\) and \(B\) with small probability.
\end{itemize}

\begin{lemma}[Completeness] \label{lem:GeomHash.Complete}
Suppose $p$ is a good prime (see \cref{lem:bad-primes-geomhash}). If \(A\) and \(B\) are congruent, then the \(\CongSign\) algorithm always outputs the same signature value on \(A\) and \(B\).
\end{lemma}

\begin{proof}
Completeness is almost immediate from the fact that every choice made in the \(\CongSign\) algorithm is canonical. In particular, if \(B=\rho A+t\) for a rotation \(\rho\) and translation \(t\), then \(\Pi'(B)=\rho^k\Pi'(A)\), and consequently the quantities \(\widehat{r},\widehat{r}',\widehat{r}_k,\widehat{r}'_k\) coincide for the two sets.

Since congruence is distance-preserving, the conditions in the \hyperlink{GeomSign 3rd pass}{third pass} of the $\CongSign$ algorithm are fulfilled by a point $z\in A$ iff the corresponding conditions are satisfied for $\rho z+t\in B$. This shows that the reference pair computed for $A$ and $B$ are indeed congruent under $\rho$ and $t$.

Finally, the signature computation in the \hyperlink{GeomSign 4th pass}{fourth pass} of the $\CongSign$ algorithm depends only on the distances and the relative positions related to the reference pair. Therefore, if $(w_1,w_2)$ is a reference pair in $A$, then its Karp-Rabin hash $F_{w_1,w_2}(A)$ is identical to that of the corresponding pair in $B$:
\[
    F_{w_1,w_2}(A) = F_{\rho w_1+t,\rho w_2+t}(B).
\]
This concludes that the final hash values $\tilde{H}(A)$ and $\tilde{H}(B)$ are identical.

For the isometry retrieval, the ordered reference points correspondence \(a\mapsto a'\), \(b\mapsto b'\) determines a unique rigid motion \((r,t)\), recovered offline by a constant-size solver. Since the matching fingerprints already certify that the remaining points align, \((r,t)\) carries \(A\) onto \(B\).
\end{proof}

\begin{lemma}[Soundness] \label{lem:GeomHash.Soundness}
Suppose $p$ is a good prime (see \cref{lem:bad-primes-geomhash}). If \(A\) and \(B\) are not congruent, then the \(\CongSign\) algorithm outputs the same hash value on \(A\) and \(B\) with probability at most $\delta/(2m^2)$. 
\end{lemma}
\begin{proof}
Let \(A\) and \(B\) be non-congruent point sets.
Notice that in the final hash values $\tilde{H}(A)$ and $\tilde{H}(B)$, the first component corresponds to the reference pair distance. Therefore, we can focus on the reference pairs $w_1^A,w_2^A\in A$ and $w_1^B,w_2^B\in B$ such that $\abs{w_1^A-w_2^A} = \abs{w_1^B - w_2^B}$.

For such a reference pair $(w_1^A,w_2^A)$, each point $z\in A$ is determined by its distances to the two references together with the side information relative to the line joining $w_1^A$ and $w_2^A$, and the same observation holds for $B$ and thus $\Gamma_{w_1^A,w_2^A}(A) \neq \Gamma_{w_1^B,w_2^B}(B)$. Since $p$ is a good prime, for $\tilde{H}(A)=\tilde{H}(B)$ to occur, the only possibility is that the two multisets are mapped to the same element $\dbF_q$ in the Karp-Rabin hash in the \hyperlink{GeomSign 4th pass}{fourth pass}. Notice that $\Gamma_{a,b}(S)$ is an $n$-element multiset from a finite domain of size at most $2^{46}U^{64}$. By choosing $q$ as a uniform random prime from the range $[2^{47}U^{64}m^2\delta^{-1}, 2^{48}U^{64}m^2\delta^{-1}]$, the collision probability for the Karp-Rabin hash on $\Gamma_{a,b}(S)$ is bounded by $\delta/(2m^2)$.
\end{proof}

We proceed to bound the space used by \(\CongSign\). Compared to the $\CongIden$ algorithm, the only additional space usage is from the reference pairs step, so it suffices to show that the reference point set is of constant size.

\begin{claim}\label{clm:anchor-size}
Suppose $p$ is a good prime (see \cref{lem:bad-primes-geomhash}). For every \(S\in\mathcal{F}\), the reference points set \(A\) maintained by \(\CongSign\) has size at most \(32\).
\end{claim}
\begin{proof}
Recall that the reference points set consists of points $a\in S$ that satisfy \cref{eq:anchor-condition}. Since $p$ satisfies Properties \ref{geomhash.prop.c} and \ref{geomhash.prop.e} of \cref{lem:bad-primes-geomhash}, this transforms to conditions in $\dbQ[\ii]$ injectively.
Fix a choice \(\rho\in\{r,r'\}\) and \(\rho_k\in\{r_k,r'_k\}\), and consider reference points satisfying
\[
|z-c_S|^2=\rho
\qquad\text{and}\qquad
\lvert (z-c_S)^k-\Pi'(S)\rvert^2=\rho_k.
\]
Let $z\in S$, and set \(u\defeq (z-c_S)^k\). Then \(|u|^2=\rho^k\) and \(|u-\Pi'(S)|^2=\rho_k\). Viewing these conditions in $\dbR^2$ for a moment, since \(\Pi'(S)\) is guaranteed to be nonzero, $u$ is defined by the intersection of two circles, hence there are at most two possible values of $u$ for fixed $\rho,\rho_k,\Pi'(S)$.

By property \ref{geomhash.prop.f}, \(\phi\) is injective on the relevant range, so for each of the two possible values of \(u\) the solutions to \(\phi(z^k)=\phi(u)\) are exactly the images of the solutions to \(z^k=u\); in particular the two solution sets have the same cardinality. By \cref{lem:niven's}, the latter equation has at most four solutions in \(\dbQ[\ii]\), hence so does the former.

Therefore, for a fixed \((\rho,\rho_k)\), there are at most \(2\cdot 4=8\) anchors, thus the reference points set has size at most \(2^2\cdot 8=32\).
\end{proof}

\begin{proof}[Proof of \cref{lem:geom-signature}]
By \cref{lem:bad-primes-geomhash,lem:GeomHash.Complete,lem:GeomHash.Soundness}, the $\CongSign$ algorithm outputs valid signatures for two sets $A,B\in \mathcal{F}$ with probability at least $1-\delta/m^2$.

For the space usage, the $\CongSign$ algorithm stores only constant elements in $\dbF_p[\ii]$ in the first two passes, thus the space usage is upper-bounded by $O(\log p)$. For the reference set computation in the third pass, as the reference set is of constant size as shown in \cref{clm:anchor-size}, $O(\log p)$ space is sufficient for this stage. For the signature generation in the forth pass, the Karp-Rabin hash $F_{a,b}(S)$ can be computed on-the-fly in $\dbF_q$ for each reference pair $(a,b)$, and accounting for the final signature, the space usage for this stage is $O\left(\log\abs{\boundedQ{2^{11}U^{16}}^2} + \log q\right) = O(\log q)$. In summary, the overall space usage of the algorithm is
\[
    \max\set{O(\log p), O(\log q)} = O\left(\log n+\log U +\log m + \log \frac{1}{\delta}\right).
\]
Finally, the length of the signature is $O(\log q) = O(\log U+\log m + \log \frac{1}{\delta})$.
\end{proof}

\section{Rational Congruence Identification in Three Dimensions}\label{sec:3D}
An immediate obstacle to extending our two-dimensional algorithms to $\dbR^3$ is the absence of a direct analogue of the correspondence between $\dbR^2$ and $\dbC$.
A dimension-reduction approach might suggest using a \emph{birthday-paradox} argument to sample a \emph{matched pair} $(a,b)\in A\times B$ with $\rho(a)+t=b$ and project the instance onto a plane, thus reducing to the two-dimensional case and achieving $O(\sqrt{n} \log n(\log n+\log U))$ space usage. However, the non-vanishing moment guarantee fails for a plane in $\dbQ^3$, as the three-dimensional integer lattice exhibits richer geometric symmetries. For example, consider the regular triangle:
\[
    S = \{(-3,1,2), (1,2,-3), (2,-3,1)\}.
\]
Under the canonical orthogonal projection of the plane $x + y + z = 0$ that the triangle lies on $\dbC$, the 3-fold rotational symmetry of the regular triangle ensures that the $m$-th moment $\sum_{z\in S} z^m$ vanishes for any $m$ not divisible by 3. Therefore, every $2^j$-th moment is identically zero, completely dismantling the non-vanishing moment guarantee under this projection. This also gives us an obstacle to directly extend the product anchor methods to the three-dimensional case, as the projected coordinates may violate the rational setting.

Employing an enhanced birthday-paradox method, we design a 6-pass randomized streaming algorithm for $\CongIden_{n,\boundedQ{U},3}$ with sublinear space complexity.
\begin{theorem}
    $\CongIden_{n, \boundedQ{U} , 3}$ can be solved by a 6-pass randomized streaming algorithm using $O(n^{5/6}\log^4 U\log n)$ space with probability at least $2/3$.
\end{theorem}

While the full algorithm and its correctness are deferred to \cref{app:3d-algo}, we provide an overview here. 
The main goal of our algorithm is to produce two families of non-collinear anchor 3-tuples, $\mathcal{P}_A \subseteq A^3$ and $\mathcal{P}_B \subseteq B^3$, each of size $O(n^{5/6})$, to test candidate transformations via the alignment of anchor points. The first anchor point is obtained by collecting points at a fixed distance from the centroids using $\dbF_p$-hashing, as in the $\CongIden_{n,\boundedQ{U},2}$ algorithm, and selecting either a precise centroid or points on a fixed-radius sphere according to \cref{obs:3dcen}. The second and third points are selected via a birthday-paradox argument.
Once $\mathcal{P}_A$ and $\mathcal{P}_B$ are obtained, the algorithm checks if there exists $(\alpha, \beta) \in \mathcal{P}_A \times \mathcal{P}_B$ such that $(\alpha,\beta)$ is a matched pair that aligns the sets $A$ and $B$. The existence can be verified efficiently due to the size guarantees of $\mathcal{P}_A$ and $\mathcal{P}_B$, and the alignment can be checked using Karp–Rabin hashing. By the birthday-paradox argument (\cref{cor:birthday}), such a matched pair exists with high probability whenever $A$ and $B$ are congruent.

\section{Hardness Results}\label{sec:hardness}
In this section, we establish lower bounds for exact and approximate geometric congruence. In particular, we show that our exact congruence algorithms achieve essentially optimal space bounds, while approximate congruence exhibits a linear pass-pace tradeoff.
Following the standard streaming literature, we prove the lower bounds via reduction from communication problems.

\subsection{Communication Complexity}
Communication complexity forms a core toolkit for establishing lower bounds in streaming models. In this section, we review the key aspects relevant to this work and refer readers to \cite{KN96,Rou16} for additional background.

For a two-player communication problem  $F:\mathcal{X}\times \mathcal{Y}\to \set{0,1}$, 
we denote $\Dcc(F)$ the deterministic communication complexity of $F$, and $\Rcc_{\delta}(F)$ the \emph{private-coin} randomized communication complexity of $F$ with error at most $\delta\in (0,1/2)$.
The following communication complexity bounds can be found in any communication complexity textbook, e.g. {\cite{KN96}}.
\begin{fact}\label{fact:eq-cost}
    Denote $\EQ_n:\set{0,1}^n\times \set{0,1}^n\to \set{0,1}$ the equality function, defined by $\EQ_n(x,y)=1$ iff $x=y$.
    Then $\Dcc(\EQ_n) = \Theta(n)$ and $\Rcc_{\delta}(\EQ_n)=\Omega(\min\set{n,\log n+\log \frac{1}{\delta}})$ for $\delta \in (0,1/2)$.
\end{fact}
\begin{fact}\label{fact:disj-cost}
    Denote $\DISJ_n:2^{[n]}\times 2^{[n]}\to \set{0,1}$ the (set) disjointness function, defined by $\DISJ_n(X,Y)=1$ iff $X\cap Y=\emptyset$.
    Then $\Dcc(\DISJ_n)=\Theta(n)$ and $\Rcc_{1/2 - \epsilon}(\DISJ_n)=\Omega(\epsilon^2 n)$ for $\epsilon \in (0,1/2)$.
\end{fact}

In most communication-complexity literature and streaming applications, public-coin and private-coin randomized models are treated as essentially equivalent. By Newman’s theorem \cite{New91}, the communication complexities of an $n$-bit function in these two randomness models differ by at most $O(\log n)$. Even though this logarithmic overhead is usually inconsequential, this distinction becomes significant for a reduction from $\EQ_n$.

Randomized streaming algorithms can be defined under several randomness models, a distinction that has received growing attention. Following the formalized randomness model framework in \cite{CS24}, the random-seed and random-tape models are the most practically relevant. In particular, the random-tape model best captures space-bounded computation: fresh randomness is freely available, but any random information reused later must be stored and counted toward the space usage.
Under the random-tape model (or the weaker random-seed model), the canonical streaming-to-communication reduction yields a private-coin communication protocol.
\begin{lemma}\label{lem:comm-to-streaming}
    Suppose a streaming problem $\mathcal{P}$ is reduced from a communication problem $F$. If $\Rcc(F)\geq c$, then any $p$-pass randomized algorithm solving $\mathcal{P}$ requires $\Omega(c/p)$ space.
\end{lemma}
\begin{proof}
    For an input $(x,y)$ to $F$, generate the data stream for $\mathcal{P}$ based on the reduction by placing all $x$-dependent items first, then all $y$-dependent items, and lastly any input-independent items. Notice that no items in the data stream can depend on both $x$ and $y$. Now Alice and Bob can create a communication protocol $\Pi$ for $F$ by simulating a streaming algorithm $\mathcal{A}$ for $\mathcal{P}$ as follows. In a pass, Alice first runs $\mathcal{A}$ until all $x$-dependent items are exhausted, and she sends the current memory state $\sigma$ to Bob. Upon reception of $\sigma$, Bob resumes the streaming algorithm $\mathcal{A}$ on the rest of the data stream, and he sends the current memory state $\sigma'$ to Alice at the end of the pass. 
    
    If $\mathcal{A}$ is a $p$-pass $s$-space algorithm, the above communication protocol involves $(2p-1)$ rounds of communication, and thus the total communication cost is at most $(2p-1)s$. 
    Moreover, when simulating the streaming algorithm, each random bit is revealed only to the currently active player. Thus, every use of randomness in $\mathcal{A}$ can be treated as private randomness in $\Pi$. In fact, $\mathcal{A}$ may need to store information from the random strings for later use by the same active player during the simulation, which can only further reduce the communication cost of $\Pi$.    
    Now, the randomized communication complexity lower bound of $c$ implies that $(2p-1)s\geq c$, hence $s=\Omega(c/p)$. 
\end{proof}

\subsection{Lower bounds for exact congruence}
In this section, we prove several lower bounds for the congruence identification problem for general dimension $d$.
We will frequently use the following crude bound on $\abs{\boundedQ{U}}$ in the reductions.
\begin{fact}
    $\abs{\boundedQ{U}} = \Theta(U^2)$.
\end{fact}

It is easy to see that the output size of $\CongIden_{n,\boundedQ{U},d}$ is already $\Omega\left(\frac{d\log U}{p}\right)$. We show that even when accounting only for the workspace, namely assuming that the output of the streaming algorithm can be written to a separate unbounded write-only space, the intermediate computation still requires the same magnitude of workspace.
\begin{proposition}\label{prop:congiden-LB}
    For $\delta\in (0,1/2)$, any $p$-pass randomized streaming algorithm for $\CongIden_{1,\boundedQ{U},d}$ with error at most $\delta$ requires a \emph{workspace} usage of $\Omega\left(\frac{d\log U}{p}\right)$.
\end{proposition}
\begin{proof}
    Denote $V=\abs{\boundedQ{U}^d}=\Theta(U^{2d})$.
    Consider the communication problem $F$ in which Alice holds an integer $x \in [V-1]$, and Bob’s goal is to determine the value of $x$. The reduction from $F$ to an instance of $\CongIden_{1,\boundedQ{U},d}$ is straightforward: order the $V$ elements of $\boundedQ{U}^d$ lexicographically as $q_1,\ldots,q_V$, and on the input $x$, Alice creates the sets $A=\set{q_1}$ and $B=\set{q_{1+x}}$.
    A simple information-theoretic argument shows that $\Rcc_\delta(F)=\Omega(\log V)=\Omega(d\log U)$; conversely, the value of $x$ can be retrieved from a correct transformation outputted by the streaming algorithm.
    Hence, \cref{lem:comm-to-streaming} implies the claimed space lower bound.
\end{proof}

\begin{proposition}\label{prop:congtest-LB}
    For $\delta\in (0,1/2)$, any $p$-pass randomized streaming algorithm for $\CongTest_{n,\boundedQ{U},d}$ with error $\delta$ requires a space usage of 
    \[
        \Omega\left(\frac{1}{p} \left(\log (n-d)+ \log\log \left(1+\frac{U^{2d}}{n-d}\right) +\log \frac{1}{\delta}\right)\right).
    \]
\end{proposition}
\begin{proof}
    Denote $W=\abs{\boundedQ{U-2}^d}=\Theta(U^{2d})$ and $k=n-d-2$.
    We present a reduction from the communication problem of equality $\texttt{EQ}:[\binom{W+k-1}{k}] \to \set{0,1}$ to $\CongTest_{n,\boundedQ{U},d}$. Order the $W$ elements of $\boundedQ{U-2}^d$ lexicographically as $q_1,\ldots,q_W$. 
    Each non-negative integral solution $\vec{v}$ to $v_1+\ldots+v_W=k$ corresponds to a size-$k$ multi-subset of $\set{q_1,\ldots,q_W}$, and by star-and-bar principle the number of such configurations is $\binom{W+k-1}{k}$.
    
    Now each integer $z\in [\binom{W+k-1}{k}]$ can be uniquely identified with a valid configuration $S(z)$. On the $\EQ$ input $(x,y)$, Alice and Bob construct $S(x)$ and $S(y)$ respectively.
    After that, the players construct the anchor multiset $F$. Let $\alpha=1/U$ and $\beta=1/(U-1)$, define
    \[
        p_j \defeq (\underbrace{\alpha,\ldots,\alpha}_{j},\underbrace{\beta,\ldots,\beta}_{d-j}) \in \boundedQ{U}^d
    \]
    for $j\in [0,d]$, and define the anchor multiset $F=\set{p_0,p_0,p_1,\ldots,p_d}$. The $\CongIden_{n,\boundedQ{U},d}$ input is $(A,B)$, where $A=S(x)\cup F$ and $B=S(y)\cup F$.

    Note that for $0\leq i<j\leq d$, $\norm{p_i-p_j}=\sqrt{j-i} \eta$, where $\eta \defeq 1/(U(U-1))$.
    It can be also verified that if $x\in \boundedQ{U-2}^d$, then for any $y\in \boundedQ{U}^d$, the distance $\norm{x-y}$ is distinct from $\norm{p_i-p_j}$ for any $i,j\in [0,d]$.
    Therefore, \(p_0\to p_1\to\cdots \to p_d\) is the unique path whose consecutive vertices are at distance $\eta$ and the starting point has multiplicity two.
    As $F$ is affine independent, this forces that $A$ is congruent to $B$ iff $A$ is identical to $B$. Therefore, the output of $\CongTest_{n,\boundedQ{U},d}(A\cup F,B\cup F)$ is exactly $\texttt{EQ}(x,y)$. By \cref{fact:eq-cost,lem:comm-to-streaming}, and the inequality $\left(\frac{N}{r}\right)^r\leq \binom{N}{r}\leq \left(\frac{eN}{r}\right)^r$, we obtain the space lower bound
    \begin{align*}        
        \Omega\left(\frac{1}{p}\left(\log\log \binom{W+k-1}{k}+\log \frac{1}{\delta}\right)\right) &= \Omega\left(\frac{1}{p}\left(\log \left(k\log \frac{W+k-1}{k}\right)+\log \frac{1}{\delta}\right)\right) \\
        &= \Omega\left(\frac{1}{p} \left(\log (n-d)+ \log\log \left(1+\frac{U^{2d}}{n-d}\right) + \log\frac{1}{\delta}\right)\right). \qedhere
    \end{align*}
\end{proof}
Combining \cref{prop:congiden-LB,prop:congtest-LB} for $d=2$ yields the streaming lower bound stated in \cref{cor:2D-congiden-LB}, hence justifying that the space usages of our algorithms are essentially optimal. 

\subsection{Hardness results for approximate congruence} \label{subsec:apx_cong_lower_bound}
We first formally define the \emph{approximate geometric congruence} problem. For $D\subseteq \dbR$, given two point multisets $A, B \subseteq D^d$ of size $n$ each, and a tolerance parameter $\epsilon > 0$, the approximate congruence testing problem (denoted as $\ApxCong_{n,D,d}^\epsilon$)is to determine whether there exists a rotation $\rho$, a translation $t$, and a bijective index map $\ell: A \to B$\footnote{The index map is required to be consistent: if $a_i,a_j\in A$ with $a_i=a_j$, then $\ell(a_i)=\ell(a_j)$.}, such that for all $a \in A$, $\norm{\ell(a) - (\rho(a) + t)} \le \epsilon$. 

We restate the desired streaming lower bound of $\ApxCong_{n,\boundedQ{U},d}^\epsilon$ for convenience:
\apxHardness*
One may use a more elaborate anchoring construction in the reduction to improve the $\frac{U}{\max\set{\epsilon,1}}$ term to $\left(\frac{U}{\max\set{\epsilon,1}}\right)^d$. We omit this refinement as it does not affect the qualitative conclusion, nevertheless we present the proof in a form that readily generalizes to higher dimensions.

\begin{proof}
    To simplify the proof, we can assume that $U\in \dbN$ and $\epsilon\in [U]\cup\set{1/r:r\in [U]}$. The other cases are handled by a suitable rounding, which will supply a lower bound of the same asymptotic order. We will prove that $\ApxCong^\epsilon_{n,\boundedQ{U},d}$ can be reduced to $\DISJ_k$ for some $k$.

    We will define a collection of sufficiently separated $\epsilon$-balls, so that for the input point sets $A$ and $B$ of $\ApxCong$ in the reduction, each ball contains exactly one point in $A$ and one point $B$. 
    Let $\bfe_1$ be the first standard basis vector,  $C=\set{c_1,\ldots,c_{k'}}$ be the set of ball centres used to encode the $\DISJ$ instance, where $c_j \defeq 5\epsilon j \bfe_1$ and
    \[
        k' = \begin{cases}
            \floor{\frac{U}{5}}-1 &\text{if }\epsilon\in \set{1/t:t\in [2,U]} \\
            \floor{\frac{U}{5\epsilon}}-1 &\text{if }\epsilon\in [U/10] 
        \end{cases}.
    \]
    Thus the size of $C$ is $k'=\Theta\left(\frac{U}{\max\set{\epsilon,1}}\right)$. We set $k = \min\set{n-2,k'}$, which is the size of the $\DISJ$ instance. We also define the anchor sets to be $F_A=\set{-f_A\bfe_1,f_A\bfe_1}$ and $F_B=\set{-f_B\bfe_1,f_B\bfe_1}$, where
    \[
        f_A = U\min\set{\epsilon,1}
        \quad\text{and}\quad
        f_B = f_A - \epsilon.
    \]

    For $b\in \set{-1,0,1}$ and $j\in [k]$, define $v_j(b)\defeq c_j +b\epsilon\bfe_1$.
    Given the inputs $x,y\in \set{0,1}^k$ to the $\DISJ_k$ instance, we construct the point sets $A$ and $B$ 
    \[
        A = \set{v_j(x_j):j\in [k]}\cup F_A
        \quad\text{and}\quad
        B = \set{v_j(-y_j):j\in [k]}\cup F_B.
    \]
    It is straightforward to verify that $v_j(z_j)\in \BB_\epsilon(c_j)$ for every $z\in \set{-1,0,1}^k$ and $j\in [k]$, and $A,B\subseteq \boundedQ{U}^d$.
    
    Suppose $A$ and $B$ are $\epsilon$-approximately congruent via a rotation $\rho$, a translation $t\in \dbR^d$, and a bijective labelling map $\ell:A\to B$.\footnote{In the one-dimensional case, the argument can be much simplified due to the fact that there are only two rotation maps over $\dbR^1$, namely the identity and the $(-1)$-multiplication.} The construction enforces that $\ell$ must map $f_A\bfe_1$ to $f_B\bfe_1$ and $-f_A\bfe_1$ to $-f_B\bfe_1$. By the approximation guarantee, we have
    \[
        \norm{f_B\bfe_1 - (\rho f_A\bfe_1 +t)} \leq \epsilon
        \quad\text{and}\quad
        \norm{-f_B\bfe_1 - (-\rho f_A\bfe_1 +t)} \leq \epsilon
    \]
    Writing $y=f_B\bfe_1 - \rho f_A\bfe_1$, the parallelogram law yields
    \begin{equation}\label{eq:pgram-law}
        \norm{y}^2+\norm{t}^2 = \frac{\norm{y-t}^2 + \norm{y+t}^2}{2} \leq \epsilon^2.
    \end{equation}
    On the other hand, by the reversed triangle inequality
    \begin{equation}\label{eq:rev-tri}
        \norm{y} \geq \absB{f_B\norm{\bfe_1} - f_A\norm{\rho\bfe_1}} = \abs{(f_A-\epsilon) - f_A} = \epsilon.
    \end{equation}
    The inequalities \cref{eq:pgram-law,eq:rev-tri} together force that $t=0$.
    
    Since for distinct $j,j'$, the balls $\BB_\epsilon(c_j)$ and $\BB_\epsilon(c_{j'})$ have distance at least $3\epsilon$, a valid bijection $\ell$ must map $v_j(x_j)\in A$ to $v_j(-y_j)\in B$ for every $j$. Since 
    \[
        \max_j \norm{v_j(x_j) - v_j(-y_j)}
        \begin{cases}
            \leq \epsilon &\text{if }\DISJ(x,y) = 1 \\
            = 2\epsilon &\text{if }\DISJ(x,y) = 0 
        \end{cases},
    \]
    it follows that $A$ and $B$ are $\epsilon$-approximately congruent iff $\DISJ_k(x,y)=1$. The desired lower bound then follows from \cref{fact:disj-cost}.
\end{proof}

\section*{Acknowledgement}
We wish to thank the anonymous reviewers for their valuable comments. In particular, we are grateful to one reviewer for pointing out the use of multiple $\ell_0$-samplers in \cref{app:3d_sampling_oracle} to extend the complex moment algorithm to the turnstile model.

Claude Opus 4.8 and GPT 5.5 were used in this work during the development stage for literature review and verification assistance. The AI tools were not used for exposition generation purposes.

\bibliographystyle{alpha}
\bibliography{ref}

\crefalias{section}{appendix}
\appendix \label{appendix}

\section{Deferred Proofs} \label{app:defered_proof}

\subsection{Proof that centroids determine the translation for congruent sets} \label{app:recentring}
More generally, for two congruent sets $A$ and $B$, any corresponding pair of anchor points $z_A$ and $z_B$ yields translated sets $A-z_A$ and $B-z_B$ that are congruent via a pure rotation. The formal statement follows immediately from the definitions of $z_A$ and $z_B$.
\begin{claim}\label{lem:anchor}
    Suppose $A,B\subseteq \dbR^d$ satisfy $B=\rho(A)+t$ for some rotation $\rho:\dbR^d\to \dbR^d$ and a translation $t\in\dbR^d$. If $z_A,z_B\in \dbR^d$ satisfy $z_B = \rho(z_A)+t$, then $A-z_A$ is congruent to $B-z_B$ through a pure rotation $\rho$.
\end{claim}
The centroids of two $n$-point congruent sets $A$ and $B$ form a pair of anchor points. Indeed, if $B=\rho A+t$, then
\[
    c_B = \frac{1}{n}\sum_{z\in B} z = \frac{1}{n}\sum_{z'\in A} (\rho z'+t) = \rho\left(\frac{1}{n}\sum_{z'\in A} z'\right) + t = \rho c_A +t.
\]

\subsection{Deferred proofs in \cref{subsec:QU}} \label{app:omitted-QU}
\begin{proof}[Proof of \cref{obs:lin-sys-precision}]
    Consider the linear system
    \[
    \begin{bmatrix}
        a_1 & a_2\\ a_3 & a_4
    \end{bmatrix}
    \begin{bmatrix}
        x \\ y
    \end{bmatrix} = \begin{bmatrix}
        a_5 \\ a_6
    \end{bmatrix}, 
    \]
    where each $a_i$ is $U$-rational. The solution is given
    \[
        x= \frac{a_4a_5-a_2a_6}{a_1a_4-a_2a_3}
        \quad 
        y=\frac{a_1a_6-a_3a_5}{a_1a_4-a_2a_3}.
    \]
    By \cref{obs:sum-prod-precision}, $x$ and $y$ are $4U^8$-rational.
\end{proof}
\begin{proof}[Proof of \cref{obs:rational_rotation}]
    Note that 
    \(
        b_2 - b_1 = (\rho a_2 +t) - (\rho a_1 +t) = \rho(a_2-a_1).
    \)
    Let $\theta,c_1,c_2,d_1,d_2\in \dbR$. Write $\rho = \cos\theta +\ii \sin\theta$, $b_2-b_1 = c_1+c_2\ii$ and $a_2-a_1 = d_1+d_2\ii$. 
    Note that by \cref{obs:sum-prod-precision}, $c_1,c_2,d_1$ and $d_2$ are $2U^2$-rational.
    We form the $2\times 2$ linear system
    \[
        \left\{\begin{aligned}
            c_1 &= d_1\cos\theta - d_2\sin\theta \\
            c_2 &= d_1\sin\theta + d_2\cos\theta
        \end{aligned}
        \right.
        \implies 
        \begin{bmatrix}
            -d_2 & d_1 \\ d_1 & d_2 
        \end{bmatrix} \begin{bmatrix}
            \sin\theta \\ \cos\theta
        \end{bmatrix} = \begin{bmatrix}
            c_1 \\ c_2
        \end{bmatrix}.
    \]
    \cref{obs:lin-sys-precision} implies that $\rho$ is $2^{10} U^{16}$-rational complex. Applying \cref{obs:sum-prod-precision} to the real and imaginary parts, we can verify that $t = b_1 - \rho a_1$ is $2^{22}U^{35}$-rational complex.
\end{proof}

\subsection{Proof of \cref{lem:norm1-modular-emb}}\label{app:norm1-modular-emb}
It suffices to show that complex conjugation is equivalent to the $p$-th power Frobenius map under the embedding into $\dbF_{p^2}$. This is a classical fact for quadratic extensions of finite fields, see, e.g., \cite[Section 2]{LN97}. Here, we give a direct argument with light number theory.

\begin{lemma}
For $p\equiv3 \bmod 8$, if $z\in \dbQ[\ii]$ is faithfully embedded in $\dbF_{p^2}$, then $\phi_p(\overline z)=\phi_p(z)^p$.
\end{lemma}

\begin{proof}
Write $z=a+b\ii$ with $a,b\in \dbQ$. Since $p=8p'+3$ for some $p'\in \dbZ$, therefore $\ii^p=\ii^{4(2p
')+3}=-\ii$ in $\dbF_p[\ii]$.
Note that $\dbF_{p^2}$ has characteristic $p$, therefore $\binom{p}{r}\pequiv 0$ for every $r\in [p-1]$. Therefore, by homomorphism property and binomial theorem,
\[
    \phi_p(z)^p
    = (\phi_p(a) + \phi_p(b) \ii)^p
    = \phi_p(a)^p+\phi_p(b)^p \ii^p = \phi_p(a)-\phi_p(b)\ii = \phi_p(a-b\ii)
\]
as claimed.
\end{proof}

\subsection{Proof of \cref{lem:M_1(T_s)}} \label{app:M(T_j)}
\cref{lem:M_1(T_s)} is a special case of the following claim with $k=0$ and $j=\Delta=s$.
    \begin{claim}\label{clm:M(T_j)}
        For each $j\in [s]$, $\Delta\in[0,j]$ and $k\in \dbN$, $M_{2^k}(T_j)$ is a polynomial in $\set{M_{2^{k+\ell}}(T_{j-\Delta})}_{\ell=0}^\Delta$ with zero constant term.
    \end{claim}
    \begin{proof}
        Let $\mathcal{P}_{j,\Delta,k}$ denote the statement to be proven.
        We prove by induction on $j$ and $\Delta$ that $\mathcal{P}_{j,\Delta,k}$ holds for every $k$. The base case $j=0$ is trivial. Suppose for some $j$, $\mathcal{P}_{j,\Delta,k}$ holds for every $\Delta \in [0,s]$ and $k\in \dbN$. Consider the statement $\mathcal{P}_{j+1,\Delta,k}$ for some $\Delta\in [0,j+1]$ and $k\in \dbN$. The sub-case $\Delta=0$ is trivial. For $\Delta\geq 1$, write $\Delta'=\Delta-1$. By \cref{eq:M-cross},
        \[
            M_{2^k}(T_{j+1}) = \frac{1}{2}(M_{2^k}(T_j)^2 - M_{2^{k+1}}(T_j)).
        \]
        Applying the induction hypotheses $\mathcal{P}_{j,\Delta',k}$ and $\mathcal{P}_{j,\Delta',k+1}$, $M_{2^k}(T_j)$ and $M_{2^{k+1}}(T_j)$ are zero-constant polynomials in $\set{M_{2^{k+\ell}}(T_{j-\Delta'})}_{\ell=0}^{\Delta'}$ and $\set{M_{2^{k+1+\ell}}(T_{j-\Delta'})}_{\ell=0}^{\Delta'}$ respectively. Therefore, $M_{2^k}(T_{j+1})$ is a polynomial in 
        \[
            \set{M_{2^{k+\ell}}(T_{j-(\Delta-1)})}_{\ell=0}^{\Delta-1} \cup \set{M_{2^{k+1+\ell}}(T_{j-(\Delta-1)})}_{\ell=0}^{\Delta-1} = \set{M_{2^{k+\ell}}(T_{j+1-\Delta})}_{\ell=0}^{\Delta}
        \]
        with zero costant term.
    \end{proof}

\subsection{Proof of \cref{clm:bad-primes-count}} \label{app:2d-bad-primes}
For Property \ref{prop.b}, $\abs{z - c_{S}}^2$ is $2^5 n^4 U^{4n+4}$-rational for each $z\in S$. Consider the set
    \[
        \Psi^{(-)} \defeq \set{d_1-d_2\in \dbR:d_1,d_2\in \Psi, d_1>d_2}.
    \]
    Each element in $\Psi^{(-)}$ is $2^{11}n^8U^{8n+8}$-rational, and $\absB{\Psi^{(-)}}\leq \binom{2n}{2} = n(2n-1)$.    
    Therefore $N_{(c)}\leq n(2n-1)$ and $U_{(c)}\leq 2^{11}n^8U^{8n+8}$.

    For Property \ref{prop.c}, we give a (very crude) bound for the precision of the $k$-th recentred complex moment $M'_{k}(L)$ for each $k\in \dbN$. For $z=x+y\ii$ where $x,y\in \boundedQ{U}$, by binomial theorem and \cref{obs:sum-prod-precision}, both the real and imaginary parts of $(x+y\ii)^k$ are $U'$-rational, where $U'\leq k\left(\binom{k}{k/2} U\right)^k\leq k 2^{k^2} U^k$, therefore the complex number $M'_{k}(L)$ is $k(k2^{k^2}U^k)^k = k^{k+1} 2^{k^3} U^{k^2}$-rational.
    Therefore $N_{(d)}\leq 4\log n+1$, and $U_{(d)}\leq \max_{k\in [n]}\set{ k^{k+1} 2^{k^3} U^{k^2} } = n^{n+1}2^{n^3}U^{n^2}$.

    For Property \ref{prop.d}, adopting a similar bound as for Property \ref{prop.c}, for each $k\in \dbN$ and $\rho\in \boundedQ{U''}[\ii]$ where $U''=2^{10}U^{16}$, $\rho^{k}$ is $U_k$-rational where $U_k=k 2^{k^2}(U'')^k\leq n2^{11n^2} U^{16n}$. Consider the set 
    \[
        \Gamma^{(-)} \defeq \set{\rho^{2^j} - (\rho')^{2^j}: \rho,\rho'\in \boundedQ{U''}[\ii]} \setminus\set{0}.
    \]
    By a similar argument for $N_{(c)}$ and $U_{(c)}$, we obtain that $N_{(e)}\leq \absB{\Gamma^{(-)}} \leq \binom{(U'')^4}{2} \leq 2^{80} U^{128}$ and $U_{(e)}\leq 2U_k^2 \leq 2n^2 2^{22n^2} U^{32 n}$.

\subsection{Proof of \cref{clm:geomhash-bad-primes-count}}\label{app:geomhash-bad-primes}
For Property \ref{geomhash.prop.d}, fix \(k\in[n]\). As in \cref{app:2d-bad-primes}, we bound the precision of \((z-c_S)^k\) and \(\Pi'(S)\). Since \(c_S\) is \(nU^n\)-rational complex, for each \(z\in S\) we can write \(z-c_S=x+y\ii\) where \(x,y\in \boundedQ{2nU^{n+1}}\). By the binomial theorem and \cref{obs:sum-prod-precision}, both the real and imaginary parts of \((z-c_S)^k = (x+y\ii)^k\) are $U'$-rational complex numbers, where 
\[
    U'\leq k\left(\binom{k}{k/2}\cdot 2nU^{n+1}\right)^k \leq k2^{k^2} (2n)^k U^{k(n+1)} \leq 2^{n^2+n} n^{n+1} U^{n(n+1)},
\]
where we use $k\leq n$.

Note that \(\Pi'(S)=\Pi_{w\in S}(w-c_S)^k\) is a product of \(n\) $U'$-rational complex numbers. Applying \cref{obs:sum-prod-precision}, for each $z \in S$, the quantity $w_{z}\defeq (z-c_S)^k - \Pi'(S)$ is a $U_{(d')}$-rational complex number, where
\begin{equation*}
    U_{(c')} \leq 2(U')^{n+1} \leq 2^{n^2+n+1}n^{n^2+2n+2} U^{n(n+1)^2}. 
\end{equation*}
Writing \(w_{j,z}=\frac{\alpha_1(z)}{\beta_1(z)}+\frac{\alpha_2(z)}{\beta_2(z)}\ii\) in lowest terms, we have
\[
    S_{(d')} \defeq \set{ \alpha_1(z),\alpha_2(z) : z\in S }\setminus\set{0},
\]
thus $N_{(d')}\leq 2n$.

For Property \ref{geomhash.prop.e}, recall that \(\Psi_{k}(S)=\set{|(z-c_S)^k-\Pi'(S)|^2 : z\in S\,}=\set{\,|w_{j}|^2 : z\in S}\). Since \(w_{z}\in \boundedQ{U_{(d')}}[\ii]\), both \((\Re (w_{z}))^2\) and \((\Im (w_{z}))^2\) are \(U_{(d')}^2\)-rational, and hence \(|w_{z}|^2=(\Re (w_{z}))^2+(\Im (w_{z}))^2\) is \(2U_{(d')}^4\)-rational. Consider the set
\[
\Psi_{k}^{(-)} \defeq \set{d_1-d_2\in \dbR : d_1,d_2\in \Psi_{k}(S),\ d_1>d_2 }.
\]
Each element in \(\Psi_{k}^{(-)}\) is \(2\cdot (2U_{(d')}^4)^2=8U_{(d')}^8\)-rational, and \(\abs{\Psi_{k}^{(-)}}\le \binom{n}{2}\). We have \(N_{(d)}\le \binom{n}{2}\leq n^2\), and 
\[
    U_{(e')}\le 8U_{(d')}^8 \leq 8\left( 2^{3n^2}n^{4n^2}U^{4n^3}\right)^8 \leq 2^{27n^2} n^{32n^2} U^{32n^3}.
\]

\section{Root Solver for Special Equations}
\subsection{Solving $\widehat{x}^k=\widehat{\lambda}$ in $G_{p+1}$}\label{app:root-recovery}
We recap the problem setting. Let $p\equiv 3\bmod 8$, then $q=(p+1)/4$ is an odd number. By Chinese Remainder Theorem, $G_{p+1}\cong G_4\times G_q$. Also, there exist $e_4,e_q\in \dbZ$ such that 
\[
    e_4+e_q=1,\quad
    \begin{cases}
        e_4 \equiv 1 \bmod 4\\ 
        e_4 \equiv 0 \bmod q
    \end{cases},\quad 
    \begin{cases}
        e_q \equiv 0 \bmod 4\\ 
        e_q \equiv 1 \bmod q
    \end{cases},
\]
which can be computed efficiently using the extended Euclidean algorithm. In particular, for any $\widehat{y}\in \dbF_{p^2}^\times$, $\widehat{y}^{e_4}\in G_4$ and $\widehat{y}^{e_q}\in G_q$.

The equation $\widehat{x}^k=\widehat{\lambda}$ in $G_{p+1}$ is equivalent to 
\[
    (\widehat{x}^{e_4+e_q})^k =\widehat{\lambda}^{e_4+e_q} \implies \widehat{x}^{ke_4} \widehat{\lambda}^{-e_4} = \widehat{x}^{-ke_q} \widehat{\lambda}^{e_q}.
\]
Now LHS of the equation represents an element in $G_4$ and RHS represents an element in $G_q$. Since $\gcd(4,q)=1$, $G_4\cap G_q=\set{1}$, which implies the $G_4$-equation $\widehat{x}^{ke_4} =\widehat{\lambda}^{e_4}$ and the $G_q$-equation $\widehat{x}^{ke_q} = \widehat{\lambda}^{e_q}$. The $G_4$-equation has at most 4 roots lying in $G_4=\set{\pm 1,\pm \ii}$. For our purposes, we simply mark all four values as the possible roots.

For the $G_q$-equation, by assumption $\gcd(k,q)=1$, therefore the equation $\widehat{y}^k = \widehat{\lambda}^{e_q}$ has a unique solution. For $\ell=k^{-1}\bmod q$, we have $\widehat{y}=\widehat{y}^{k\ell} = \widehat{\lambda}^{\ell e_q} \eqdef \widehat{w}$. With $\widehat{x}^{e_4}\in \set{\pm 1,\pm \ii}$ the substitution $\widehat{y}=\widehat{x}^{\ell_q}$, this gives the four candidate solutions $\set{\pm \widehat{w}, \pm \ii\widehat{w} }$ to the original equation. 

\subsection{Solving $\widehat{x}^{2^j}=\widehat{\lambda}$ in $\dbF^\times_{p^2}$ for $p\equiv 3 \bmod 16$}\label{app:power-of-two-root}
By strengthening the condition for $p$ to $p\equiv 3\bmod 16$, we can ensure that the number of roots of the equation $\widehat{x}^{2^j}=\widehat{\lambda}$ in $\dbF^\times_{p^2}$ is at most $\gcd(2^j, p^2-1) \leq 8$.
Now such an equation can be solved by a similar Chinese Remainder Theorem decomposition approach detailed in \cref{app:root-recovery}.

By Chinese Remainder Theorem, $\dbF^\times_{p^2} \cong G_8\times G_q$, where $q\defeq (p^2-1)/8$ is odd. It suffices to solve the corresponding equations in the $G_8$ and $G_q$ parts individually. The $G_q$-equation yields a unique root $\widehat{w}$, which can be found by inverting $2^j$ in $\dbZ_q$ and raising $\widehat{\lambda}$ to that exponent. 

For the $G_8$ component, although a primitive 8th root of unity $\widehat{\zeta}$ is not immediately available as in $\cref{app:root-recovery}$, it can be efficiently constructed using a simple sampling procedure. The sampler chooses $a$ uniformly at random from $\dbF_{p^2}^{\times}$, and set $\widehat{\zeta}\defeq a^{(p^2-1)/8}$. It returns $\widehat{\zeta}$ if $\widehat{\zeta}^4\neq 1$, and otherwise resamples. As the element $\widehat{\zeta}$ is uniformly distributed over $G_8$, each sample yields a primitive 8th root of unity with probability $1/2$, and hence $\widehat{\zeta}$ can be constructed efficiently in both time and space. Finally, the solution set of the original equation is $\set{\widehat{\zeta}^j \widehat{w}:j\in [8]}$.

\section{Adaptions in Turnstile Model}
\subsection{Hardness of Finding the Minimum Element in Turnstile Model}
\label{app:min_element}
We show that in the turnstile model, finding the minimum element is hard even when multiple passes are allowed. We show this by a reduction from the \emph{minimum missing element} problem: given a stream of elements forming a set \(A\), with the promise that every element smaller than \(\textsf{MME}(A) = \min\{i \geq 0 : i \notin A\}\) occurs exactly once, compute \(\textsf{MME}(A)\).
Guha and McGregor~\cite[Theorem 5]{GuhMcG08}, any \(p\)-pass insertion-only randomized algorithm for this problem requires \(\widetilde{\Omega}(n^{1/p})\).
Given a streaming algorithm for the minimum element, we simulate each pass by first inserting every element of \(\{0, \dots, n\}\), generated on-the-fly by a counter, and then deleting each element of \(A\) as it arrives in the stream.
The resulting stream leaves exactly the complement \(\{0, \dots, n\} \setminus A\), whose minimum is \(\textsf{MME}(A)\).
This simulation preserves the number of passes and introduces an overhead of only \(O(\log n)\) bits, so finding the minimum element in \(p\) passes also requires \(\widetilde{\Omega}(n^{1/p})\) space.
In particular, no constant-pass algorithms could achieve polylogarithmic space.

\subsection{Complex Moment Method Adaptation for the Turnstile Model} \label{app:turnstile}

One can verify that all quantities involved remain computable in the \emph{turnstile model}; the only adjustment necessary is for selecting a surviving non-centroid datum during the second pass.

In the turnstile model, the first observed datum may be deleted later. Instead, we use $O(\log 1/\delta)$ independent $\ell_0$-samplers~\cite{10.1145/1989284.1989289} in parallel, updating their sketches with every insertion and deletion. At the end of the first pass, we query these samplers to retrieve a uniform random sample from data stream. Each sampler succeeds with probability $1 - \eta$ and returns a distinct element from the stream with near-uniform probability $1/N \pm \eta$, where $N$ is the number of distinct elements of the underlying set and $\eta \le 1/N^2$. This ensures we capture a valid non-centroid point with high probability:

\begin{claim}
    Let $A \subseteq D$ be a multiset with $N \ge 2$ distinct elements. If there exist $z \in A$ such that $z \neq p$ for a fixed $p \in D$, then $O(\log \frac{1}{\delta})$ parallel $\ell_0$-samplers will return an element $z \in A$ such that $z \neq p$ with probability at least $1 - \delta$.
    \begin{proof}
        For a single $\ell_0$-sampler, a failure occurs if it returns nothing, or if it returns $p$. The probability of this failure is bounded by $\eta + (1 - \eta)(1/N + \eta) \le c$ for some constant $c \in (0,1)$. For the failure probability of all $T$ independent $\ell_0$-samplers to be bounded by $\delta$, it suffices to take $c^T \leq \delta$, i.e. $T=\Theta(\log 1/\delta)$. 
    \end{proof}
\end{claim}

Each $\ell_0$-sampler requires space usage of \(O(\log^2 |D| \log 1/\eta)\). By letting $D = \boundedQ{U}^2$ and $\eta = 1/n^2$ (recall that $n$ is the size of the input multisets), the overall space complexity of the turnstile-model complex-moment algorithm for $\CongIden_{n,\boundedQ{Q},2}$ is summarized as follows.
\begin{corollary}\label{cor:CM-turnstile}
    For $\delta\in (0,1)$,
    $\CongIden_{n,\boundedQ{U},2}$ can be solved by a 3-pass $O(\log n(\log n+\log^2 U\cdot \log \frac{1}{\delta}))$-bit randomized turnstile streaming algorithm with probability at least $1-\delta$.
\end{corollary}

\section{$\CongIden_{n,\boundedQ{U},3}$ Algorithm} \label{app:3d-algo}
For the three-dimensional setting, let \(\tilde{S}\) denote the set of distinct elements underlying a multiset \(S\). We use the same notation $\mathcal S_r(c)$ to denote a radius-$r$ sphere from the centre $c\in\mathbb R^3$, i.e.
\(
\mathcal S_r(c) = \{ x \in \mathbb R^3 : \|x - c\|_2 = r \}.
\)

Furthermore, we employ the oracle \(\mathcal{R}(k, \epsilon)\) to sample \(k\) distinct elements from the data stream. This sampling is nearly uniform and has a success probability at least \(1 - \epsilon\). The oracle uses \(O(k\log^4U\log\frac{1}{\epsilon})\) bits of space, and we further elaborate on this part in \cref{app:3d_sampling_oracle}.

Our $\CongIden_{n,\boundedQ{U},3}$ algorithm proceeds in two stages: \emph{canonical form computations} and \emph{equality testing}. The main technical work lies in the canonical form computations, where we process each of the sets $A$ and $B$ in parallel during the streaming phase. For the remainder of this section, we focus on a single set $S$ (either $A$ or $B$) and describe its canonical form computations.

For a given $S$, we iteratively build a candidate set of triples $\mathcal{P}_S = \{v_1^S,v_2^S,v_3^S\}$. 
Every triple produced in this process is non-coplanar. Among all constructed triples, at least one satisfies that each $(v_i^{A},v_i^{B})$ is a matched pair under the hidden congruence.
For each candidate triple, we build a canonical hash from these anchors and use it to test congruence.

The actual canonical form computations for the set $S$ in our algorithm are as follows:
\begin{enumerate}
    \item \emph{First reference point $v_1^S$}
    
    While $(c_A,c_B)$ is a matched pair, just as in the two-dimensional case, the centroids may be $nU^n$-rational and cannot be stored with full precision. We need the following observation.

    \begin{observation}\label{obs:3dcen}
    For a rational point set $S\subseteq \boundedQ{U}^3$ and any fixed radius $r$, exactly one of the following holds:
    \begin{enumerate}[label=\textup{(\alph*)}]
        \item $\abs{\calB_r(c_S) \cap \tilde{S}} < 4$;
        \item $\abs{\calB_r(c_S) \cap \tilde{S}} \ge 4$ and the points in $\calB_r(c_S) \cap \tilde{S}$ are coplanar;
        \item $\abs{\calB_r(c_S) \cap \tilde{S}} \ge 4$ and the points in $\calB_r(c_S) \cap \tilde{S}$ are not coplanar.
    \end{enumerate}
    \end{observation}

    We handle the three cases as follows:
    In Case~(a), we include each of the at most four distinct points in $\calB_r(c_S) \cap \tilde{S}$ as a candidate for $v_1^S$ to $\mathcal{P}_S$.
    In Case~(b), any three coplanar points in $\calB_r(c_S) \cap \tilde{S}$ determine the centre of their circle, which is a $U^{O(1)}$-rational point; this centre is added as a candidate for $v_1^S$ to $\mathcal{P}_S$.
    In Case~(c), any four non-coplanar points in $\calB_r(c_S) \cap \tilde{S}$ determine the centre of their sphere, again a $U^{O(1)}$-rational point; this centre is added as a candidate for $v_1^S$ to $\mathcal{P}_S$. 
    
    Choosing $r$ and constructing $\calB_r(c_S) \cap \tilde{S}$ also requires careful field operations as in the two-dimensional case; we defer these details to \cref{app:3d_correctness}.
    
    \item \emph{Second reference point $v_2^S$}
    
    For each candidate $v_1^S \in \mathcal{P}_S$, let $r=\|q^S - v_1^S\|$, where $q^S$ is the first datum in $S$ distinct from $v_1^S$. If no such $q^S$ exists, then \(\abs{\tilde{S}} = 1\) and the problem is solved trivially. 
    We then use \(\mathcal{R}(O(\sqrt{N}), n^{-1})\) to sample \(O(\sqrt{n})\) distinct points from $\calB_r(v_1^S)\cap S$ and denote this sample by $S$.
    For every $v_2^S \in S$, we replace the candidate $\{v_1^S\}$ with the extended pair $\{v_1^S,v_2^S\}$.
        
    \item \emph{Third reference point $v_3^S$}\label{step.3}
    
    Now for each candidate pair $\{v_1^S, v_2^S\}\in \mathcal{P}_S$, let $r = \|q^S - v_1^S\|$, where $q^S$ is the first point in $L$ that is not collinear with $v_1^S, v_2^S$.
    If no such $q^S$ exists, then all points of $L$ lie on the line through $v_1^S$ and $v_2^S$, we extend each candidate by letting \(v_3^S = v^S_{\perp} + v_1^S\), where \(v^S_{\perp}\) is the vector with fixed length that is orthogonal to \(v^S_2 - v^S_1\), and proceed directly to the next pass.
    
    Consider the sphere $\calB_r(v_1^S) \cap S$ and denote its size \(m = \abs{\calB_r(v_1^S) \cap S}\).
    To select $v_3^S$, we partition $\calB_r(v_1^S)\cap S$ according the slices. For $k\in \dbR$, we define
    \[
        \Pi_k \defeq \set{v\in\dbR^3: (v_2^S - v_1^S)\cdot(v - v_1^S) = k }
    \]
    and $C_k \defeq ((\calB_r(v_1^S)\cap S) \setminus \{v_2^S, 2v_1^S-v_2^S\}) \cap \Pi_k$.
    We call a slice $C_k$ \emph{sparse} if $\abs{C_k} \le n^{2/3}$. Among all sparse slices, we choose the one with the smallest $\abs{k}$, i.e., the slice closest to $v_1^S$.
    If two $C_k$ have the same $\abs{k}$, we choose the slice that lies on the same side of $v_1^S$ as $v_2^S$.
    
    We then sample $O(n^{1/3})$ distinct points \(v^S_3\) from this slice $C_k$ using the oracle \(\mathcal{R}(O(n^{1/3}), n^{-1})\), and for each sampled point $v_3^S$ we replace the candidate $\{v_1^S, v_2^S\}$ with the extended tuple $\{v_1^S, v_2^S, v_3^S\}$ and insert it back into $\mathcal{P}_S$.
    
    If no sparse slice exists (that is, every slice has size greater than \(n^{2/3}\)), we discard the candidate \(\{v_1^S, v_2^S\}\) and do not generate any third reference point for it.
    
    \item \emph{Canonical Hash computation}
        
    For each candidate triple $\vec{t}=\{v_1^S, v_2^S, v_3^S\} \in \mathcal{P}_S$, we compute a canonical hash for every point $v \in S$ based on its distances and orientation relative to the reference triple:
    \[
    \theta_S(v) \;=\; \left(\,\|v - v_1^S\|^{2},\; \|v - v_2^S\|^2,\; \|v - v_3^S\|^{2},\; \sigma_{\vec{t}}(v)\,\right),
    \]
    where
    \[
    \sigma_{\vec{t}}(v) \;=\; \mathrm{sgn}\!\left((v - v_1^S)\cdot\left((v_2^S - v_1^S)\times(v_3^S - v_1^S)\right)\right) \in\set{-1,0,+1}
    \]
    indicates which side of the reference plane the point \(v\) lies on. Note that \(\theta_S(v)\) is invariant under rotation, which we will prove in \cref{app:3d_correctness}.
\end{enumerate}

For the equality testing phase, we iterate over all combinations of candidate triples from $A$ and $B$.
For each pair $\{v_1^A,v_2^A,v_3^A\}\in \mathcal{P}_A$ and $\{v_1^B,v_2^B,v_3^B\}\in \mathcal{P}_B$, we first check whether the two triples are congruent.
If the pairs are not congruent, we skip this pair.
If they are congruent, then the canonical hash $\theta_S(v)$ together with its reference triple uniquely determines each point of $L$, and the triple induces the hash set
    \[
    H_S = \{\theta_S(v) : v \in S \}.
    \]
    
    We then test whether $H_A = H_B$.
    This is done by computing a Karp-Rabin fingerprint for each candidate and comparing the two fingerprints.
    If the fingerprints match, we have found a valid correspondence between $A$ and $B$.
    In this case, the matching anchor triples determine the rotation $\rho$ and translation $t$, and we output the transformation $(\rho,t)$.

\subsection{Space and Correctness of the Algorithm} \label{app:3d_correctness}
We first give a proof of \cref{obs:3dcen}.
\begin{proof}[Proof of \cref{obs:3dcen}]
    Notice that any four non-coplanar points determine a unique sphere. Suppose there exists a sphere \(\calB\) centred at the centroid \(C\) and the points on the sphere are non-coplanar. Let \(v_i = (x_i, y_i, z_i)\in \dbR^3\) for $i\in [4]$ be four non-coplanar points on \(\calB\), then for $i\in\set{2,3,4}$    
\[
\norm{v_i - C}^2 = \norm{v_1 - C}^2 \Rightarrow 2C \cdot (v_i-v_1) = v_i \cdot v_i - v_1 \cdot v_1 = \norm{v_i}^2-\norm{v_1}^2.
\]
This yields a linear system $MC = b$, where    
\[ 
M = \begin{bmatrix}
    x_2 - x_1 & y_2 - y_1 & z_2 - z_1 \\
    x_3 - x_1 & y_3 - y_1 & z_3 - z_1 \\
    x_4 - x_1 & y_4 - y_1 & z_4 - z_1 \\
\end{bmatrix}, \text{ and }
b = \frac{1}{2}
\begin{bmatrix}
   x_2^2 + y_2^2 + z_2^2 - (x_1^2 - y_1^2 - z_1^2) \\
   x_3^2 + y_3^2 + z_3^2 - (x_1^2 - y_1^2 - z_1^2) \\
   x_4^2 + y_4^2 + z_4^2 - (x_1^2 - y_1^2 - z_1^2)
\end{bmatrix}
\]

As \(M\) is invertible, \(C\) is uniquely determined. By Cramer's rule, one has
\[
C = \left(\frac{\det(M_1)}{\det(M)}, \frac{\det(M_2)}{\det(M)}, \frac{\det(M_3)}{\det(M)}\right),
\]
where \(M_i\) is the matrix formed by replacing the \(i\)-th column of \(M\) by the column vector \(b\). 
\end{proof}

We give a remark on the precision of $C$ in this case.
Notice that each entry of \(M\) is \(2U^2\)-rational, and each entry of \(b\) is \(18U^{12}\) rational. For any matrix \(X \in \mathbb{Q}^{3\times3}\) with each entry being \(V\)-rational, its determinant is \(12V^{27}\)-rational. By simple calculations, each coordinate of \(C\) is $2^{150}U^{378}$-rational.

Now, we formally state and prove the \emph{birthday paradox} argument used in our algorithm:
\begin{lemma}[{Birthday Paradox}]
\label{cor:birthday}
Let $A = \{a_1, \dots, a_N\}, B = \{b_1, \dots, b_N\} \subseteq \boundedQ{U}^3$ be congruent set, with \(a_i\) matches \(b_i\) for each $i$. For  $m = O(\sqrt{kN})$, suppose $S_A \subseteq A$ and $S_B \subseteq B$ are size-$m$ subsets sampled using independent sampling oracles \(\mathcal{R}(O(\sqrt{k\abs{A}}), \epsilon)\). Then, with probability at least \(5/6\), there exist $k$ matched pairs $(a_i,b_i)\in S_A\times S_B$ with all $a_i$'s distinct and all $b_i$'s distinct.
\end{lemma}

\begin{proof}
Let \(m = c\sqrt{kN}\), where \(c\) is a constant determined later. We define the random variable $Z$ as the number of selected matched pairs, i.e. the number of $i$'s in which $(a_i,b_i)\in S_A\times S_B$. Then \(Z = \sum_{i = 1}^N I^A_iI^B_i\), where \(I^A_i\) is the indicator variable for the event \(a_i \in S_A\), and \(I^B_i\) is the indicator for \(b_i \in S_B\).

We now estimate \(\mu = \E[Z]\). First, for each $i\in [N]$ and $S\in\set{A,B}$, the expected value of $I^S_i$ is bounded below by 
\[
\E[I^S_i]  = \Pr(I^S_i = 1) \ge 1 - \left(1 - \left(\frac{1}{N} - \frac{1}{64U^{12}} \right) \right)^m \ge 1 - \left(1 - \frac{1}{2N} \right)^m \ge 1 - e^{-m/2N} \ge \frac{m}{4N}.
\]

Using the inequality \(1 - e^{-x} \ge x/2\) for all \(0 \le x \le 1\), we have \(\E[I^S_i] \ge {m}/(4N)\). Similarly,
\[
\E[I^S_i] \le 1 - \left(1 - \left(\frac{1}{N} + \frac{1}{64U^{12}}\right)\right)^m \le 1 - \left(1 - \frac{2}{N}\right)^m \le \frac{2m}{N}.
\]

Since \(S_A\) and \(S_B\) are independently generated, \(\E[Z] = \sum_{i=1}^N\E[I^A_i]\E[I^B_i]\) is bounded by
\[
\frac{c^2k}{16} = N \cdot \left(\frac{m}{4N}\right)^2 \le \E[Z] \le N \cdot \left(\frac{2m}{N}\right)^2 = 4c^2k.
\]

Next, we claim that \(\Var(Z) \le \mu\). The variance is given by
\[
\Var(Z) = \sum_{i = 1}^N \Var(I^A_iI^B_i) + \sum_{i \neq j} \text{Cov}(I^A_iI^B_i, I^A_jI^B_j).
\]
Note that for an indicator variable \(X\), \(\Var(X) = \E[X](1 - \E[X]) \le \E[X]\). Thus, \(\sum_i \Var(I^A_i I^B_i) \le \sum_i \E[I^A_i I^B_i] = \mu\).
Moreover, as the events of picking distinct matched pairs \((a_i, b_i)\) and \((a_j, b_j)\) are negatively correlated. Thus, the covariance terms are non-positive. Consequently, \(\Var(Z) \le \mu\).

Finally, we use Chebyshev's inequality to bound the probability that we find fewer than \(k\) matched pairs:
\begin{align*}
        \Pr(Z < k) &\le \Pr(\abs{Z - \mu} \ge \mu - k) \le \frac{\Var(Z)}{(\mu - k)^2}
    \le \frac{\mu}{(\mu - k)^2}
    \le \frac{4c^2k}{\left(\frac{c^2k}{16} - k\right)^2} \leq \frac{1}{6}.
\end{align*} 
where the last inequality holds for any $k\in \dbN$ when $c=80$. Thus, we have bounded the success probability as claimed.
\end{proof}

Now we justify the pass usage and correctness of our $\CongIden_{n,\boundedQ{U},3}$ algorithm. 
\begin{enumerate}
\item \emph{First reference point \(v_1^S\) (1st and 2nd pass)}
\begin{itemize}
    \item 1st Pass:

    We select a prime \(p\) from the first \(K\) primes uniformly at random, where \(K\) will be determined later, and compute the reduced radius \(\phi_p(r^2)\). Here, the reduction map \(\phi_p: \mathbb{Q} \to \mathbb{F}_p\) is defined by \(a/b \mapsto a \cdot b^{-1} \pmod{p}\).

    \item 2nd Pass:

    After obtaining the radius \(\phi_p(r^2)\), we calculate the distance of each point \(v \in S\) from \(c_S\) and form the set 
    \[
    \tilde{\calB}_{r, \phi_p}(c_S) = \{v \in \tilde{S} \lvert \phi_p((v - c_S)^2) = \phi_p(r^2)\}
    \]
    Next, first reference points \(v_1^S\) will be determined according to the set \(\tilde{\calB}_{r, \phi_p}(c_S)\).
    
\end{itemize}

To ensure the probability that \(\tilde{\calB}_{r, \phi_p}(c_S) = \calB_r(c_S) \cap \tilde{S}\), we first upper-bound the number of \emph{bad primes} \(p\) --prime numbers for which the reduction modulo \(p\) fails to distinguish the distance \(r\) from \(c_S\) between others.

\begin{lemma}\label{app:3dbadprime}
    The number of primes \(p\) for which the selection process fails to correctly identify the points on the sphere \(\calB_r(c_S)\) is \(O(n^2 \log U)\).
\end{lemma}
\begin{proof}
    The process fails if the map is undefined or if it yields a false positive (collisions).
    
    To ensure that \(\phi_p\) is well-defined, \(p\) must not divide the denominator of any coordinate involved in the point set \(L\). As the coordinates of the points in \(L\) are bounded by \(U\), each denominator has at most \(\log U\) prime divisors. Thus, there are at most \(n\log U\) primes that divide any denominator, making the map undefined. 

    A false positive occurs when there exist \(x \in S\) such that \(\phi_p(||x-c_S||^2) = \phi_p(r^2)\) but \(x \notin \calB_r(c_S)\). This is equivalent to $\phi_p(\norm{x - c_S}^2) \pequiv 0$.
    Since the numerator of \(\norm{x - c}^2\) is at most \(O\left(U^{8(n + 1)}\right)\), the total number of primes that could cause a false positive is \(O(n^2 \log U)\). 

    Combining both cases, the total count of \emph{bad primes} is \(O(n^2 \log U)\).
\end{proof}

Let \(K = O(n^3\log U)\), by the Prime Number Theorem, the magnitude of the \(K\)-th prime is approximately \(K\log K\). If we select a prime \(p\) from the first \(K\) primes uniformly at random, the probability of choosing a \emph{bad prime} is bounded by \(O(1/n)\). Thus, the process succeeds with high probability \(1 - O(1/n)\) while ensuring space efficiency.

\item \emph{Second reference point \(v_2^S\) (3rd pass)}

Since the first reference point \(v_1^S\) can be stored in \(O(\log U)\) bits, we can directly calculate the distance between each point \(p\) and \(v_1^S\), then sample points \(v_2^S\) from \(\calB_r(v_1^S) \cap \tilde{S}\). 

The space usage of this pass is dominated by the sampling oracle \(R(O(n^{1/2}, n^{-1}))\), which is \(O(\sqrt{n} \log^4U\log n)\) bits.

\item \emph{Third reference point \(v_3^S\) (4th and 5th pass)}
\begin{itemize}
    \item 4th pass:
    
In this pass, we aim to find the slice \(C_k\) for sampling the third reference points \(v_3^S\). Recall that for a fixed \(\{v_1^S, v_2^S\}\), we seek the slice with the smallest absolute offset \(k\) satisfying the sparsity condition \(\lvert C_k \rvert \le n^{2/3}\). Here, we call \(v_2^S\) a \emph{bad} point with respect to \(v_1^S\) if there is no such sparse slice. Notice that:

\begin{observation}
The total number of dense slices (i.e., \(\lvert C_k \rvert > n^{2/3}\)) is at most \(m / n^{2/3} \le n^{1/3}\). 
\end{observation}

Therefore, if we examine the sequence of slices ordered by their distance from \(v_1^S\), the target sparse slice must be found within the first \(n^{1/3}\) distinct slices. If all such slices are dense, then \(v_2^S\) is classified as a \emph{bad} point.

Consequently, the algorithm does not need to track all potential slices. For each candidate pair \(\{v_1^S, v_2^S\}\), we only maintain counters for the first \(n^{1/3}\) slices closest to \(v_1^S\). In this pass, we stream the points and update these counters to determine the correct index \(k\) (or discard \(v_2^S\) as bad).

\item 5th pass:

We sample \(O(n^{1/3})\) points on the specific slice \(C_k\).

\end{itemize}

For each first and second candidate points \(\{v_1^S, v_2^S\}\), it requires \(O\left(n^{1/3}\log^3 U\log n\right)\) bits to determine the slice \(C_k\) and sample \(v_3^S\). Hence, the total space usage is 
\[O\left(\sqrt{n}\right) \cdot O\left(n^{1/3}\log^3 U\log n\right) = O\left(n^{5/6}\log^4 U\log n\right).\]

Next, we establish that the algorithm generates matched reference triples with constant probability. Before that, we first argue that there are only a few second candidates \(v_2^S\) that will be discarded. The lemma follows:

\begin{lemma}\label{app:badpoint}
    For any fixed \(v_1^S\), there exist at most two \emph{bad points} with respect to \(v_1^S\) on every sphere centred at \(v_1^S\).
\end{lemma}
    
\begin{proof}
    For \(v \in \dbR^3 \setminus \set{v_1^S}\), let \(\tilde{\calB}_v = \calB_r(v_1^S) \cap \tilde{S}\) generated by the candidate pair \(\{v_1^S, v\}\), \(m_v = \lvert \tilde{\calB}_v \rvert\), and \(\mathcal{H}_v\) be the set of planes which are orthogonal to \(v - v_1^S\) and intersect with some \(C_k\). Notice that if \(p\) is a bad point with respect to \(v_1^S\), then
    \[
        \lvert \mathcal{H}_v \rvert \leq m_v / n^{2/3} \le  m_v^{1/3}.
    \] 
    Suppose for the sake of contradiction that there exist two bad points \(v, v'\) equidistant to \(v_1^S\), and \(v, v', v_1^S\) are non-collinear. Without loss of generality, assume that \(m_v \le m_{v'}\). Since the normal vectors \(n_v = v - v_1^S\) and \(n_{v'} = v' - v_1^S\) are not parallel, any two plane \(H \in \mathcal{H}_v, H' \in \mathcal{H}_{v'}\) are not parallel, and their intersection \(H \cap H'\) is a line. The set of all such intersection lines is \(\mathcal{L} = \{H\cap H' \vert H \in \mathcal{H}_v, H' \in \mathcal{H}_{v'}\}\) and its size is bounded by:
    \[
    \lvert \mathcal{L} \rvert = \lvert \mathcal{H}_v \rvert \cdot \lvert \mathcal{H}_{v'} \rvert \le (m_{v'}m_v)^\frac{1}{3} \le m_{v'}^{2/3}
    \]
    All the points in \(\tilde{\calB}_v \cup \tilde{\calB}_{v'}\) (except for the ignored points \(\{v, v', 2v_1^S - v, 2v_1^S - v'\}\)) lie on one of the line in \(\mathcal{L}\). Furthermore, every line in \(\mathcal{L}\) intersect  \(\tilde{\calB}_v \cup \tilde{\calB}_{v'}\) in at most four points, since \(\tilde{\calB}_v \cup \tilde{\calB}_{v'}\) is the union of two spheres and each \(\tilde{\calB}_v, \tilde{\calB}_{v'}\) contain no duplicates. We can thus bound \(m_{v'}\) by counting the points on these lines:
    \[
    m_{v'} \le \lvert \tilde{\calB}_v \cup \tilde{\calB}_{v'} \rvert \le 4\lvert \mathcal{H}_v \rvert \cdot \lvert \mathcal{H}_{v'} \rvert + 4 = O(m_{v'}^{2/3}).
    \]
    The inequality is a contradiction for sufficiently large \(m_{v'}\). Therefore, the assumption that non-collinear bad points exist is false. The only possible bad points on the sphere must be collinear with \(v_1^S\), so the only candidates are \(p\) and its antipode \(2v_1^S - v\).
\end{proof}

\begin{lemma}\label{app:3dprob}
Suppose $A$ and $B$ are congruent. Then, with probability at least \({25}/{36} - O(n^{-1})\), \(\mathcal{P}_A\) contains at least a triple \(\vec{t}\) that matches a triple pair \(\vec{t}_* \in \mathcal{P}_B\).
\end{lemma}

\begin{proof}
By our construction, the candidates for the first reference point \(v_1^S\) are determined deterministically. Thus, the set of candidates for \(v_1^A\) is guaranteed to map to the set of candidates for \(v_1^B\).

Let \((v_1^A, v_1^B)\) be a pair of fixed matched points, then the sphere \(\calB_r(v_1^A) \cap \tilde{A}\) also matches \(\calB_r(v_1^B) \cap \tilde{B}\). If we sample \(80\sqrt{3n}\) points as the second reference points, according to the Birthday Paradox (\cref{cor:birthday}), the probability that there exists at least three matched pairs of the second candidate points \((v_2^A, v_2^B)\) is bounded below by \(5/6\). Even though some of \(v_2^S\)'s will be discarded, \cref{app:badpoint} shows that we will only remove at most two points. Therefore, the probability of finding a non-bad matched pair \((v_2^A, v_2^B)\) remains at least \(5/6 - O(1/n)\), where the term \(O(1/n)\) comes from the failure probability of the oracle \(\mathcal{R}(O(\sqrt{n}), n^{-1})\).

Conditioned on having a matched pair \(\{v_1^A, v_2^A\}\) and \(\{v_1^B, v_2^B\}\), if no point in $A$ is not collinear with $v_1^A, v_2^A$, then the generated point \(v^A_3 = v^A_{\perp} + v_1^A\) must match \(v^B_3 = v^B_{\perp} + v_1^B\). Otherwise, the slice \(C_k\) corresponding to \(\{v_1^S, v_2^S\}\) determined by the algorithm must also match. Again, by the Birthday Paradox (\cref{cor:birthday}), there is probability at least \(5/6 - O(n^{-1})\) that at least one of the matched selected points \((v_3^A, v_3^B)\) can be found for a matched pair of candidates \(\{v^A_1, v_2^A\}, \{v_1^B, v_2^B\}\).  

Since the selection steps are independent conditioned on the previous outcomes, the total probability of finding a matching triple \(\vec{t} = (v_1^A, v_2^A, v_3^A)\) and \(\vec{t}_* = (v_1^B, v_2^B, v^B_3)\) is at least \((5/6 - O(n^{-1}))^2 = 25/36 - O(n^{-1})\).
\end{proof}

Next, we show that the canonical hashing sets \(H_A, H_B\) generated by the matched triples \(\{v_1^A, v_2^A, v_3^A\}\) and \(\{v_1^B, v_2^B, v_3^B\}\) suffice to identify the congruence between \(A, B\).

\begin{lemma}\label{app:hashcorrectness}
    Let \(\vec{t} = \{v_1, v_2, v_3\}\) and \(\vec{t}_* = \{v^*_1, v^*_2, v^*_3\}\) be non-collinear triples such that \(\rho(v_i) + b = v^*_i,\ i = 1, 2, 3\). Then \(A\) is congruent to \(B\) via the rotation \(\rho\) and translation \(b\) if and only if \(H_A = H_B\).
    \end{lemma}
    
    \begin{proof}
    If \(A\) is congruent to \(B\) via the rotation \(\rho\) and translation \(b\) , since the transformation \(v \mapsto \rho(v) + b\) is isometric, for any points \(v \in A\) and \(i \in [3]\), \(||v - v_i|| = ||(\rho(v) + b) - (\rho(v_i) + b)|| = ||(\rho(v) + b) - v^*_i||\). For the orientation, notice that 
    \[(v - v_1)\cdot\left((v_2 - v_1)\times(v_3 - v_1)\right) = \det(M),\quad\text{where } M = \begin{bmatrix}
        v \\
        v_3 - v_1\\
        v_2 - v_2
    \end{bmatrix} \in \dbR^{3 \times 3}.
    \]
    Hence 
    \begin{equation}\label{ev^*_ori}
    \sigma_r(\rho(v) + b) = \mathrm{sgn}(\det(MR_\rho^\top))= \mathrm{sgn}(\det(M)\det(R_\rho^\top))= \mathrm{sgn}(\det(M))= \sigma_t(v),
    \end{equation} 
    where \(R_\rho\) is the matrix form of the rotation \(\rho\). Thus, \(H_A = H_B\).

    On the other hand, suppose that \(H_A = H_B\). Then for any \(v \in A\), there exist \(v^* \in B\) such that \(\theta_A(v) = \theta_B(v^*)\). Let \(x = \rho (v) + b\), we aim to prove that \(x = v^*\).

    Since \(\theta_A(v) = \theta_B(v^*)\), the distances to the anchor points are preserved:
    \[
    \norm{v - v_i} = \norm{v^* - v^*_i} 
    \quad \text{ for all } i \in [3].
    \]
    Since the transformation \(v \mapsto \rho (v) + b\) is an isometry, we also have:
    \[
    \norm{T(v) - T(v_i)} = \norm{x - v^*_i} = \norm{p - v_i}
    \quad \text{ for all }i \in [3].
    \]
    This implies both \(x\) and \(v^*\) lie on the intersection of the three spheres 
    \[S_i \defeq \{y \in \dbR^3 : \norm{y - v^*_i} = \norm{v^* - v^*_i}\} \quad\text{ for }i\in [3]\]. 

    Since \(v^*_1, v^*_2, v^*_3\) are non-collinear, they define a unique plane \(H\). The intersection of three spheres centred on \(H\) contains at most two points. Furthermore, if the intersection contains two distinct points, they are reflections of each other across the plane \(H\). Thus, either \(v^* = x\) or \(v^*\) is the reflection of \(x\) across \(H\).

    If the intersection of \(S_1, S_2, S_3\) contains only one point, then \(x\) is exactly \(v^*\). Otherwise, if the three spheres intersect at two points, similar to \cref{ev^*_ori}, one has \(\sigma_{\vec{t}_*}(x) = \sigma_{\vec{t}}(v) = \sigma_{\vec{t}_*}(v^*)\). This concludes that \(v^* = x\).
\end{proof}

\item \emph{Canonical Hashing (6th pass)}

In the final pass, we verify the set equality of \(H_A, H_B\) using Karp-Rabin hash. As each coordinate of \(\theta_S(v)\) is \(24U^6\)-rational for any \(v \in \boundedQ{U}^3\), we can map rational coordinates to integers via a canonical injection \(\textsf{idx}: \boundedQ{24U^6}^4 \rightarrow [(24U^6)^8]\). Now, we pick a prime \(q\) and choose \(x \in \mathbb{F}_q\) uniformly at random to compute the finger print \(h_A, h_B\) of the set \(H_A, H_B\) as
\[
h_A = \sum_{v \in H_A} x^{\textsf{idx}(v)} \mod q,\quad h_B = \sum_{v \in H_B} x^{\textsf{idx}(v)} \mod q.
\]

As the degree of \(h_A - h_B\) is at most \((24U^6)^8\), it suffice to pick \(q = \Theta(nU^{48})\) to achieve \(O\left(\frac{1}{n}\right)\) error probability. The space requirement of each fingerprint is \(O(\log U + \log n)\) bits. Hence, the space usage to store all the fingerprints is \(O\left(n^{5/6}(\log U + \log n)\right)\) bits.

Conditioned on the triples \(\{v_1^A, v_2^A, v_3^A\}, \{v_1^B, v_2^B, v_3^B\}\) are congruent and \(H_A = H_B\), by \cref{app:hashcorrectness}, the corresponding transformation \(\rho, t\) between \(A\) and \(B\) is exactly the transformation between two triples.

To summarize, our algorithm uses \(O\left(n^{5/6}\log^4 U\log n\right)\) bits of space to identify congruence:
\begin{itemize}
    \item If \(A\) is congruent to \(B\), the algorithm returns a valid rotation \(\rho\) and translation \(t\) with probability at least \((\frac{25}{36} - O(1/n)) \cdot (1 - O(1/n))^2 > 2/3\), where the factor \(1 - O(1/n)\) comes from the field operations and fingerprint hashing collision.
    \item If \(A\) is not congruent to \(B\), the algorithm returns \(\perp\) with probability \(1 - O(1/n)\).
\end{itemize}

\end{enumerate}

\subsection{Sampling Oracle for $\CongIden_{n,\boundedQ{U},3}$ Algorithm} \label{app:3d_sampling_oracle}

The oracle is constructed using independent \(\ell_0\)-samplers. A single \(\ell_0\)-sampler attempts to return a random element from a data stream of points from the domain \(\boundedQ{U}^3\). It succeeds with probability at least \(1 - \delta\), where \(\delta = \frac{1}{64U^{12}} \le \abs{\boundedQ{U}^3}^{-2}\). Conditioned on successful sampling, it outputs an element in the data stream with probability \(\frac{1}{N} \pm \delta\), where \(N\) is the number of distinct elements in the data stream.
Now, we claim that to sample \(k\) distinct elements with probability at least \(1 - \epsilon\), it suffices to use \(O(k\log U\log\frac{1}{\epsilon})\) \(\ell_0\)-samplers:

\begin{lemma}
    \(\mathcal{R}(k, \epsilon)\) can be implemented using \(O(k\log U\log\frac{1}{\epsilon})\) independent instances of \(\ell_0\)-sampler.

    \begin{proof}
        Let \(\mathcal{W}\) be a sampler that fails with probability \(1/2\) and, conditioned on success, returns a random element uniformly from the stream (i.e., probability \(1/(2N)\) for any element).
   
        Recall that an \(\ell_0\)-sampler picks any specific distinct element \(x\) with probability at least \((1 - \delta) \cdot (\frac{1}{N} - \delta)\), since \(N \le \abs{\boundedQ{U}^3}\), we have:

        \[
            (1 - \delta) \cdot \left(\frac{1}{N} - \delta\right) \ge \left(1 - \frac{1}{N^2}\right) \cdot \left(\frac{1}{N} - \frac{1}{N^2}\right) \ge \frac{1}{2N}.
        \]

        Since the real sampler is more likely to return a new distinct element than \(\mathcal{W}\) is, the number of trials required by the real oracle is stochastically dominated by the number of trials \(T\) required by an oracle built from \(\mathcal{W}\). Thus, it suffices to upper bound the number of samplers \(\mathcal{W}\) needed.

        Denote \(T\) the number of \(\mathcal{W}\) needed to collect \(k\) distinct elements. Let \(T_i\) be the number of samples needed to find the \((i+1)\)-th distinct element after finding \(i\) elements. The probability of success in a single trial is \(v_i = \frac{N-i}{2N}\). Thus, \(T_i \sim \text{Geom}(v_i)\) and the expected value of \(T\) satisfies
        \[
            \E[T] = \sum_{i=0}^{k-1} \frac{1}{v_i} = \sum_{i=0}^{k-1}\frac{2N}{N - i}.
        \]
        It is obvious that $\E[T] = O(k)$ when $k\leq N/2$. For $k> N/2$, 
        \[
            \E[T] \leq 2N\sum_{j=1}^N \frac{1}{j} \leq 2N(\log N+1).
        \]
        As $\log N = O(\log U)$, we have $\E[T] = O(k\log U)$ in any case.

        By a tail bound for the sum of independent geometric variables (see \cite[Corollary 2.4]{Jan18}), $\Pr(T \ge \lambda\E[T]) \le e^{1 - \lambda}$ for any $\lambda\geq 1$.
        Setting \(\lambda \geq \log \frac{1}{\epsilon} + 1\), we ensure that the failure probability is at most \(\epsilon\).
        Thus, the total number of samplers required is at most \(\lambda \E[T] = O\left(k \log U \log \frac{1}{\epsilon}\right).\)
    \end{proof}
\end{lemma}

Since each \(\ell_0\)-sampler uses \(O(\log^3U)\) bits of space, so the total space usage of \(\mathcal{R}(k, \epsilon)\) is \(O(\log^3U) \cdot O(k\log U\log \frac{1}{\epsilon}) = O(k\log^4U\log\frac{1}{\epsilon})\) bits.

\end{document}